\documentclass{article}%
\usepackage{amssymb}
\usepackage{amsmath}
\usepackage{amsfonts}
\usepackage{graphicx}%
\providecommand{\U}[1]{\protect\rule{.1in}{.1in}}
\newtheorem{theorem}{Theorem}
\newtheorem{acknowledgement}[theorem]{Acknowledgement}

\newtheorem{corollary}[theorem]{Corollary}

\newtheorem{definition}[theorem]{Definition}

\newtheorem{proposition}[theorem]{Proposition}
\newtheorem{remark}[theorem]{Remark}

\newenvironment{proof}[1][Proof]{\noindent\textbf{#1.} }{\ \rule{0.5em}{0.5em}}
\begin{document}

\title{Boundary problems for three-dimensional Dirac operators and generalized MIT
bag models for unbounded domains}
\author{V.S. Rabinovich\\Instituto Polit\'{e}cnico Nacional, ESIME Zacatenco, M\'{e}xico,\\email: vladimir.rabinovich@gmail.com}
\date{}
\maketitle

\begin{abstract}
We consider the operators of boundary problems
\begin{equation}
\mathbb{D}_{\boldsymbol{A,}\Phi,\mathfrak{B}}\boldsymbol{u}=\left\{
\begin{array}
[c]{c}%
\mathfrak{D}_{\boldsymbol{A},\Phi}\boldsymbol{u}\text{ on }\Omega\\
\mathfrak{B}\boldsymbol{u}_{\partial\Omega}\text{ on }\partial\Omega
\end{array}
\right.  \label{1}%
\end{equation}
in unbounded domains $\Omega\subset\mathbb{R}^{3}$ where $\mathfrak{D}%
_{\boldsymbol{A},\Phi}$ is a $3-D$ Dirac operator
\begin{align}
\mathfrak{D}_{\boldsymbol{A},\Phi}  &  =\left[  \boldsymbol{\alpha\cdot
}(i\boldsymbol{\nabla}+\boldsymbol{A})+\alpha_{0}m+\Phi I_{4}\right] \\
&  =\sum_{j=1}^{3}\alpha_{j}(i\partial_{x_{j}}+A_{j})\alpha_{0}m+\Phi
I_{4}\nonumber
\end{align}
defined on the distributions $\boldsymbol{u}=(u_{1},u_{2},u_{3},u_{4})\in
H^{1}(\Omega,\mathbb{C}^{4})$, where $\alpha_{0},\alpha_{1},\alpha_{2}%
,\alpha_{3}$ are the Dirac matrices, $\boldsymbol{A}\in L^{\infty}%
(\Omega,\mathbb{R}^{3})$ and $\Phi\in L^{\infty}(\mathbb{R}^{3})$ are the
magnetic and electrostatic potentials, $m\in\mathbb{R}$ is the mass of the
particlle. Let $\mathbb{C}^{4}\ni\boldsymbol{u}=(\boldsymbol{u}^{1}%
,\boldsymbol{u}^{2})\in\mathbb{C}^{2}\oplus\mathbb{C}^{2}.$ We assume that the
operator $\mathfrak{B}$ of the boundary condition is
\begin{equation}
\mathfrak{B}\boldsymbol{u}_{\partial\Omega}=\mathfrak{b}_{1}\boldsymbol{u}%
_{\partial\Omega}^{1}+\mathfrak{b}_{2}\boldsymbol{u}_{\partial\Omega}^{2}
\label{2}%
\end{equation}
where $\mathfrak{b}_{j},j=1,2$ are $2\times2$ matrices, $\boldsymbol{u}%
_{\partial\Omega}^{j}\in H^{1/2}(\partial\Omega,\mathbb{C}^{2}),j=1,2$ are
restrictions of distributions $\boldsymbol{u}^{j}\in H^{1}(\Omega
,\mathbb{C}^{2})$ on $\partial\Omega.$ Class of the boundary condition
(\ref{2}) in particular case contains the boundary conditions of MIT bag model
and its generalizations which discribe the confinment of the quarks to the
domain $\Omega.$

We give conditions of self-adjointnes of unbounded operators $\mathcal{D}%
_{\boldsymbol{A,}\Phi,\mathfrak{B}}$ associated with boundary problem
(\ref{1}) and to give a description of the essential spectra of $\mathcal{D}%
_{\boldsymbol{A,}\Phi,\mathfrak{B}}$ for some unbounded domains applying the
limit operators method.



\end{abstract}

\section{ Introduction and notations}

Massive relativistic particles of spin-$\frac{1}{2}$ such as electrons and
quarks are described by the Dirac operator which is the base of relativistic
quantum mechanics (see for instance \cite{Hecht},\cite{Thaller}). \ The
boundary problems for $3-D$ Dirac operators arise in relativistic quantum
mechanics for the descriptions of the particles confined in domains in
$\mathbb{R}^{3}.$ One important model of the confinment of particles is
the\ three-dimensional MIT bag model suggested in the 1970s by physicists (see
for instance \cite{Chodos},\cite{Chodos1},\cite{Jonson}). In the dimension
two, the Dirac operators with special boundary conditions similar to the MIT
bag model are used in the description of graphene \cite{Benguria},
\cite{Benguria1}, see also \cite{Stockmeyer}.

The rigorous mathematical investigation of the MIT bag model as a boundary
value problem for the $3D-$ Dirac operator in domain $\Omega\subset
\mathbb{R}^{3}$ with the special boundary conditions has been started in the
papers \cite{Ariza},\cite{Ariza1},\cite{Qurm} where the authors apply to the
study of self-adjointness of the MIT bag model the modern techniques of
extension theory of symmetric operators. We note also the very recent paper
\cite{Behrndt} devoted to the spectral theory of boundary value problems for
$3D$-Dirac operator with the MIT bag model boundary conditions and some of its
generalizations. The approach of this paper is based on abstract boundary
triple techniques from extension theory of symmetric operators and a thorough
study of certain classes of (boundary) integral operators, that appear in a
Krein type resolvent formula. In this paper the authors study also the
relations between the $3D-$Dirac operators with singular $\delta-$shell
interactions and the MIT bag model. \ 

It should be noted that there is an extensive literature devoted to the
Fredholm theory and index of boundary problems for multi-dimensional Dirac
type operators in bounded domains (see for instance \cite{BB}, \cite{BB1},
\cite{Boos}, \cite{Boos-book}, \cite{Guilkey} and references cited there).

We consider in the paper operators of boundary value problems for the
$3D-$Dirac operators with variable magnetic and electrostatic potentials,
containing in a particular case the operators of MIT bag models and some of
its generalizations. Our study the spectral properties of mentioned operators
is based on the investigation of the parameter-dependent boundary problems for
the Dirac operators and the limit operators method.

We consider the operators of boundary problems
\begin{equation}
\mathbb{D}_{\boldsymbol{A,}\Phi,\mathfrak{B}}=\left\{
\begin{array}
[c]{c}%
\mathfrak{D}_{\boldsymbol{A},\Phi}\boldsymbol{u}\text{ on }\Omega\\
\mathfrak{B}\boldsymbol{u}_{\partial\Omega}\text{ on }\partial\Omega
\end{array}
\right.  \label{0.1}%
\end{equation}
in domains $\Omega\subset\mathbb{R}^{3}$ with $C^{2}-$uniformly regular
boundary $\partial\Omega$ (the definition of uniformly regular manifold see in
\cite{Amman}) where $\mathfrak{D}_{\boldsymbol{A},\Phi}$ is a $3-D$ Dirac
operator
\begin{align}
\mathfrak{D}_{\boldsymbol{A},\Phi}  &  =\left[  \boldsymbol{\alpha\cdot
}(i\boldsymbol{\nabla}+\boldsymbol{A})+\alpha_{0}m+\Phi I_{4}\right]
\label{0.2}\\
&  =\sum_{j=1}^{3}\alpha_{j}(i\partial_{x_{j}}+A_{j})\alpha_{0}m+\Phi
I_{4}\nonumber
\end{align}
defined on the distributions $\boldsymbol{u=}(u_{1},u_{2},u_{3},u_{4})\in
H^{1}(\Omega,\mathbb{C}^{4})=H^{1}(\Omega)\otimes\mathbb{C}^{4}$ where
$H^{1}(\Omega)$ is the Sobolev space on $\Omega$ of order $1.$ We denote by
$\alpha_{0},\alpha_{1},\alpha_{2},\alpha_{3}$ the Dirac matrices satisfying
the relations
\[
\alpha_{i}\alpha_{j}+\alpha_{j}\alpha_{i}=2\delta_{ij}I_{4},i,j=0,1,2,3,4,
\]
$I_{n}$ is the $n\times n$ unit matrix, $\boldsymbol{A}=(A_{1},A_{2},A_{3})\in
L^{\infty}(\mathbb{R}^{3},\mathbb{C}^{3})$ is a magnetic potential, $\Phi\in
L^{\infty}(\mathbb{R}^{3})$ is an electrostatic potential, $m\in\mathbb{R}$ is
the mass of the particle. We use the system of coordinates for which the
Planck constant $\mathfrak{h}=1,$ the light speed $c=1,$ and the charge of the
particle $e=$ $1.$

We introduce the boundary operators of the form
\begin{equation}
\left(  \mathfrak{B}\boldsymbol{u}\right)  (x^{\prime})=\mathfrak{b}%
_{1}\mathfrak{(}x^{\prime})\boldsymbol{u}_{\partial\Omega}^{1}(x^{\prime
})+\mathfrak{b}_{2}\mathfrak{(}x^{\prime})\boldsymbol{u}_{\partial\Omega}%
^{2}(x^{\prime})\boldsymbol{,}x^{\prime}\in\partial\Omega\label{0.3}%
\end{equation}
where $\boldsymbol{u=}\left(
\begin{array}
[c]{c}%
\boldsymbol{u}^{1}\\
\boldsymbol{u}^{2}%
\end{array}
\right)  \in\mathbb{C}^{4},\boldsymbol{u}^{j}\in\mathbb{C}^{2},j=1,2,$
$\boldsymbol{u}_{\partial\Omega}^{j}\in H^{1/2}(\partial\Omega,\mathbb{C}%
^{2})$ is the trace of $\boldsymbol{u}^{j}\in H^{1}(\Omega,\mathbb{C}^{2})$ on
$\partial\Omega,$ and $\mathfrak{b}_{j}\mathfrak{=(}b_{j}^{kl})_{k,l=1}%
^{2},l=1,2$ are $2\times2$ matrices with elements $b_{j}^{kl}$ belonging to
the space $C_{b}(\partial\Omega)$ of bounded continuous functions on
$\partial\Omega.$ $\qquad$

The class of boundary conditions (\ref{0.3}) contains, in particular case, the
boundary condition of the \textit{MIT bag model} \ (see for instance
\cite{Ariza}, \cite{Ariza1}, \cite{Chodos}, \cite{Chodos1}, \cite{Jonson}),
\begin{equation}
\left(  I_{4}+i\alpha_{0}(\boldsymbol{\alpha\cdot\nu)}\right)  \boldsymbol{u}%
_{\partial\Omega}=0 \label{0.4}%
\end{equation}
where $\boldsymbol{\nu}=(\nu_{1},\nu_{2},\nu_{3})$ is the unit outward normal
vector to $\partial\Omega.$ One can see that the rank of the matrix
$I_{4}+i\alpha_{0}(\boldsymbol{\alpha\cdot\nu)}$ is $2,$ and the boundary
condition (\ref{0.4}) is equivalent to the condition
\begin{equation}
\boldsymbol{u}^{1}+\left(  i\sigma\cdot\boldsymbol{\nu}\right)  \boldsymbol{u}%
^{(2)}=0,\sigma\cdot\boldsymbol{\nu}=\sigma_{1}\nu_{1}+\sigma_{2}\nu
_{2}+\sigma_{3}\nu_{3} \label{0.5}%
\end{equation}
having the form (\ref{0.3}) where $\sigma_{1},\sigma_{2},\sigma_{3}$ are the
$2\times2$ Pauli matrices. Moreover, the boundary condition
\begin{equation}
\theta\left(  I_{4}+i\alpha_{0}(\boldsymbol{\alpha\cdot\nu)}\right)
\boldsymbol{u}_{\partial\Omega}=\left(  I_{4}+i\alpha_{0}(\boldsymbol{\alpha
\cdot\nu)}\right)  \alpha_{0}\boldsymbol{u}_{\partial\Omega} \label{0.5'}%
\end{equation}
introduced in the paper \cite{Behrndt} are equivalent to the condition%
\[
(\theta-1)\boldsymbol{u}^{1}+(\theta+1)i(\sigma\cdot\nu)\boldsymbol{u}%
^{2}=\boldsymbol{0}%
\]
of form (\ref{0.3}).

We associate with boundary value problem (\ref{0.1}) a bounded operator
\begin{equation}
\mathbb{D}_{\boldsymbol{A},\Phi,\mathfrak{B}}\boldsymbol{u(}x)=\left\{
\begin{array}
[c]{c}%
\mathfrak{D}_{\boldsymbol{A},\Phi}\boldsymbol{u}(x),\text{ }x\in\Omega\\
\mathfrak{B}\boldsymbol{u}_{\partial\Omega}(x^{\prime})=\mathfrak{b}%
_{1}\mathfrak{(}x^{\prime})\boldsymbol{u}_{\partial\Omega}^{1}(x^{\prime
})+\mathfrak{b}_{2}\mathfrak{(}x^{\prime})\boldsymbol{u}_{\partial\Omega}%
^{2}(x^{\prime})=\boldsymbol{0,}x^{\prime}\in\partial\Omega
\end{array}
\right.  \label{0.6}%
\end{equation}
acting from $H^{1}(\Omega,\mathbb{C}^{4})$ into $L^{2}(\Omega,\mathbb{C}%
^{4}).$ We also associate with the operator $\mathbb{D}_{\boldsymbol{A,}%
\Phi,\mathfrak{B}}$ the unbounded in $L^{2}(\Omega,\mathbb{C}^{4})$ operator
$\mathcal{D}_{\boldsymbol{A,}\Phi,\mathfrak{B}}$ generated by the $3D-$ Dirac
operator $\mathfrak{D}_{\boldsymbol{A},\Phi}$ with the domain
\begin{equation}
H_{\mathfrak{B}}^{1}(\Omega,\mathbb{C}^{4})=\left\{  \boldsymbol{u}\in
H^{1}(\Omega,\mathbb{C}^{4}):\mathfrak{B}\boldsymbol{u}_{\partial\Omega
}=0\text{ on }\partial\Omega\right\}  . \label{0.7}%
\end{equation}

The paper has two principal aims: $\boldsymbol{1)}$ to obtain conditions for
the operator $\mathcal{D}_{\boldsymbol{A,}\Phi,\mathfrak{B}}$ to be a
self-adjoint operator in $L^{2}(\Omega,\mathbb{C}^{4})$ for bounded and
unbounded domains $\Omega,$ and $\boldsymbol{2)}$ to study the essential
spectrum of $\mathcal{D}_{\boldsymbol{A,}\Phi,\mathfrak{B}}$ for unbounded
domains $\Omega.$

$\boldsymbol{1)}$ Our approach to self-adjointness of operators $\mathcal{D}%
_{\boldsymbol{A,}\Phi,\mathfrak{B}}$ is based on the study of invertibility of
the parameter-dependent boundary problems
\[
\mathbb{D}_{\boldsymbol{A},\Phi,\mathfrak{B}}(i\mu)=\mathbb{D}_{\boldsymbol{A}%
,\Phi,\mathfrak{B}}-i\mu I_{4},\mu\in\mathbb{R}%
\]
for large value $\left\vert \mu\right\vert $ of the parameter $\mu
\in\mathbb{R}.$ We introduce the local parameter-dependent Lopatinsky-Shapiro
conditions for the operator $\mathbb{D}_{\boldsymbol{A},\Phi,\mathfrak{B}%
}(i\mu)$ and prove the following result.

\begin{itemize}
\item Let: $(i)$ $\partial\Omega$ be a $C^{2}-$uniformly regular surface,
$(ii)$ the parameter-dependent Lopatinsky-Shapiro conditions for
$\mathbb{D}_{\boldsymbol{A},\Phi,\mathfrak{B}}(i\mu)$ be satisfied uniformly
on $\partial\Omega,$ $(iii)$ the unbounded operator $\mathcal{D}%
_{\boldsymbol{A,}\Phi,\mathfrak{B}}$ with domain $H_{\mathfrak{B}}^{1}%
(\Omega,\mathbb{C}^{4})$ be a symmetric operator in $L^{2}(\Omega
,\mathbb{C}^{4}).$ Then $\mathcal{D}_{\boldsymbol{A,}\Phi,\mathfrak{B}}$ is a
self-adjoint operator in $L^{2}(\Omega,\mathbb{C}^{4}).$
\end{itemize}

In distinction from the above cited papers \cite{Behrndt}, \cite{Ariza},
\cite{Ariza1}, \cite{Qurm} where the boundary problems have been consider for
domains with compact $\mathbb{C}^{2}$-boundary our approach allows to prove
self-adjointness of wide class of boundary value problems for $3D-$Dirac
operators in domains with bounded and unbounded boundary, in particular case,
for the operators of the MIT bag models and their generalizations.

$\boldsymbol{2)}$ We study also the Fredholm property of operators
$\mathbb{D}_{\boldsymbol{A},\Phi,\mathfrak{B}}:H^{1}(\Omega,\mathbb{C}%
^{3})\rightarrow L^{2}(\Omega,\mathbb{C}^{3})$ and the essential spectrum of
associated unbounded operators $\mathcal{D}_{\boldsymbol{A,}\Phi,\mathfrak{B}%
}$ in unbounded domains in $\mathbb{R}^{3}$ with $C^{2}-$uniformly regular
boundary applying the limit operators method \cite{RRS},\cite{Ra1},\cite{Ra2}.
We consider the following cases:

\begin{itemize}
\item $(a)$ $\Omega$ is the exterior of a bounded domain, that is
$\Omega=\mathbb{R}^{3}\diagdown\overline{\Omega^{\prime}}$ where
$\Omega^{\prime}$ is a bounded domain with a $C^{2}-$boundary $\partial
\Omega.$

\item $(b)$ $\Omega$ is an unbounded domain with $\ C^{2}-$ boundary which has
a conic exit at infinity, that is $\Omega$ coincides with a conic set outside
a ball $B_{R}=\left\{  x\in\mathbb{R}^{3}:\left\vert x\right\vert <R\right\}
,R>0.$
\end{itemize}

\ Following to the papers \cite{Ra1},\cite{Ra2} we define a family
$Lim(\mathbb{D}_{\boldsymbol{A,}\Phi,\mathfrak{B}})$ of limit operators
$\mathbb{D}_{\boldsymbol{A,}\Phi,\mathfrak{B}}^{h}$ for $\mathbb{D}%
_{\boldsymbol{A,}\Phi,\mathfrak{B}}$ and obtain the following result.

\begin{itemize}
\item Let the standard Lopatinsky-Shapiro conditions for $\mathbb{D}%
_{\boldsymbol{A,}\Phi,\mathfrak{B}}$ are satisfied uniformly on $\partial
\Omega.$ Then the essential spectrum $sp_{ess}\mathcal{D}_{\boldsymbol{A,}%
\Phi,\mathfrak{B}}$ is defined by the formula%
\begin{equation}
sp_{ess}\mathcal{D}_{\boldsymbol{A,}\Phi,\mathfrak{B}}=%
{\displaystyle\bigcup\limits_{\mathbb{D}_{\boldsymbol{A,}\Phi,\mathfrak{B}%
}^{h}\subset Lim(\mathbb{D}_{\boldsymbol{A,}\Phi,\mathfrak{B}})}}
sp\mathcal{D}_{\boldsymbol{A,}\Phi,\mathfrak{B}}^{h} \label{0.8}%
\end{equation}
where $\mathcal{D}_{\boldsymbol{A,}\Phi,\mathfrak{B}}^{h}$ are the unbounded
operators associated with $\mathbb{D}_{\boldsymbol{A,}\Phi,\mathfrak{B}}^{h}.$
\end{itemize}

Note that in the case $(a)$ the limit operators $\mathbb{D}_{\boldsymbol{A,}%
\Phi,\mathfrak{B}}^{h}$ are the Dirac operators on $\mathbb{R}^{3}$, and for
the case $(b)$, the limit operators are operators of boundary value problems
for the half-spaces. Under some additional assumptions the spectra of limit
operators can be obtained in an explicit form, that is formula (\ref{0.8})
gives the complete description of $sp_{ess}\mathcal{D}_{\boldsymbol{A,}%
\Phi,\mathfrak{B}}.$ \ 

Let $\Omega=\mathbb{R}^{3}\diagdown\overline{\Omega^{\prime}}$ be an exterior
of the bounded domain $\Omega^{\prime}$ \ with $C^{2}$-boundary $\partial
\Omega.$ We assume that the potentials $\boldsymbol{A}$ and $\Phi$ are
real-valued and slowly oscillating at infinity (see Definition \ref{de5.3})
then the essential spectrum of the operator $\mathcal{M}_{\boldsymbol{A,}%
\Phi,\mathfrak{M}}$ of MIT bag model defined by the Dirac operator
$\mathfrak{D}_{\boldsymbol{A},\Phi}$ with domain%
\[
dom\mathcal{M}_{\boldsymbol{A,}\Phi,\mathfrak{M}}=\left\{  \boldsymbol{u}\in
H^{1}(\Omega,\mathbb{C}^{4}):\mathfrak{M}\boldsymbol{u=\ u}_{\partial\Omega
}^{1}+\left(  i\sigma\cdot\nu\right)  \boldsymbol{u}_{\partial\Omega}%
^{(2)}=0\text{ on }\partial\Omega\right\}
\]
is given as
\begin{equation}
sp_{ess}\mathcal{M}_{\boldsymbol{A,}\Phi,\mathfrak{M}}=(-\infty,M_{\Phi}%
^{\sup}-\left\vert m\right\vert ]%
{\displaystyle\bigcup}
[M_{\Phi}^{\inf}+\left\vert m\right\vert ,+\infty) \label{0.9}%
\end{equation}
where
\begin{equation}
M_{\Phi}^{\sup}=\limsup_{x\rightarrow\infty}\Phi(x),M_{\Phi}^{\inf}%
=\liminf_{x\rightarrow\infty}\Phi(x). \label{0.10}%
\end{equation}
Thus $sp_{ess}\mathcal{M}_{\boldsymbol{A,}\Phi,\mathfrak{M}}$ is independent
of the slowly oscillating magnetic potential $\boldsymbol{A,}$ and if
$M_{\Phi}^{\sup}-M_{\Phi}^{\inf}<2\left\vert m\right\vert $%
\[
sp_{dis}\mathcal{M}_{\boldsymbol{A,}\Phi,\mathfrak{M}}\subset(M_{\Phi}^{\sup
}-\left\vert m\right\vert ,M_{\Phi}^{\inf}+\left\vert m\right\vert )
\]
and if $M_{\Phi}^{\sup}-M_{\Phi}^{\inf}\geq2\left\vert m\right\vert $ then%

\begin{equation}
sp\mathcal{M}_{\boldsymbol{A,}\Phi,\mathfrak{M}}=sp_{ess}\mathcal{M}%
_{\boldsymbol{A,}\Phi,\mathfrak{M}}=(-\infty,+\infty). \label{0.11}%
\end{equation}

If $\Omega$ is a domain with a conic exit to infinity and the potentials
$\boldsymbol{A}$ and $\Phi$ are real-valued and slowly oscillating at
infinity, then: if $m>0$ then the essential spectrum of $\mathcal{M}%
_{\boldsymbol{A,}\Phi,\mathfrak{M}}$ is given by formulas (\ref{0.9}%
),(\ref{0.10}), and if $m\leq0$ then the $sp\mathcal{M}_{\boldsymbol{A,}%
\Phi,\mathfrak{M}}=sp_{ess}\mathcal{D}_{\boldsymbol{A},\Phi,\mathfrak{B}%
}=(-\infty,+\infty).$

Note that the approach based on the investigation of parameter-dependentt
operators and the limit operators has been applied earlier for the
investigation of the $\delta-$interaction problems for Schr\"{o}dinger
operators in papers \cite{Rabin},\cite{Rabin1},\cite{Rabin2}.

\subsection{Notations}

\begin{itemize}
\item If $X,Y$ are Banach spaces then we denote by $\mathcal{B(}X,Y)$ the
Banach space of bounded linear operators acting from $X$ into $Y$, and by
$\mathcal{K}(X,Y)$ its subspace \ of compact operators. In the case $X=Y$ we
write shortly $\mathcal{B}(X),\mathcal{K}(X).$

\item An operator $A\in\mathcal{B(}X,Y)$ is called a Fredholm operator if%

\[
k\emph{er}A=\left\{  x\in X:Ax=0\right\}  \text{ and }\emph{coker}%
A\emph{=}Y/\operatorname{Im}A
\]
are finite dimensional spaces.

\item Let $\mathcal{A}$ be a closed unbounded operator in a Hilbert space
$\mathcal{H}$ with a dense in $\mathcal{H}$ domain $dom\mathcal{A}.$ Then
$\mathcal{A}$ is called a Fredholm operator if $\mathcal{A}:dom\mathcal{A}%
\rightarrow\mathcal{H}$ is a Fredholm operator as a bounded operator where
$dom\mathcal{A}$ is equipped by the graph norm
\[
\left\Vert u\right\Vert _{dom\mathcal{A}}=\left(  \left\Vert u\right\Vert
_{\mathcal{H}}^{2}+\left\Vert \mathcal{A}u\right\Vert _{\mathcal{H}}%
^{2}\right)  ^{1/2}%
\]
(see for instance \cite{BirSol}).

\item The essential spectrum $sp_{ess}\mathcal{A}$ of an unbounded operator
$\mathcal{A}$ is a set of $\lambda\in\mathbb{C}$ such that $\mathcal{A}%
-\lambda I$ is not a Fredholm operator as the unbounded operator, and the
discrete spectrum $sp_{dis}\mathcal{A}$ of $\mathcal{A}$ is the set of
isolated eigenvalues of finite multiplicity. It is well known that if
$\mathcal{A}$ is a self-adjoint operator then $\ sp_{dis}\mathcal{A}%
\mathfrak{=}sp\mathcal{A}\mathfrak{\diagdown}sp_{ess}\mathcal{A}\mathfrak{.}$

\item We denote by $L^{2}(\Omega,\mathbb{C}^{d}),L^{2}(\partial\Omega
,\mathbb{C}^{d})$ the Hilbert spaces of $d-$dimensional vector-functions
$\boldsymbol{u}(x)=(u^{1}(x),...,u^{d}(x)),x\in\Omega$ or $x\in\partial\Omega$
with the scalar products
\begin{align*}
\left\langle \boldsymbol{u},\boldsymbol{v}\right\rangle _{L^{2}(\Omega
,\mathbb{C}^{d})}  &  =\int_{\Omega}\boldsymbol{u}(x)\cdot\boldsymbol{v}%
(x)dx,\\
\left\langle \boldsymbol{\varphi},\boldsymbol{\psi}\right\rangle
_{L^{2}(\partial\Omega,\mathbb{C}^{d})}  &  =\int_{\partial\Omega
}\boldsymbol{\varphi}(x^{\prime})\cdot\boldsymbol{\psi}(x^{\prime})dx^{\prime}%
\end{align*}
where $\boldsymbol{h\cdot g=}\sum_{j=1}^{d}h_{j}\bar{g}_{j}.$

\item We denote by $H^{s}(\mathbb{R}^{n},\mathbb{C}^{d})$ the Sobolev space of
distributions $\boldsymbol{u}\in\mathcal{D}^{\prime}(\mathbb{R}^{n}%
,\mathbb{C}^{d})=\mathcal{D}^{\prime}(\mathbb{R}^{n})\otimes\mathbb{C}^{d}$
such that
\[
\left\Vert \boldsymbol{u}\right\Vert _{H^{s}(\mathbb{R}^{n},\mathbb{C}^{d}%
)}=\left(  \int_{\mathbb{R}^{n}}(1+\left\vert \xi\right\vert ^{2}%
)^{s}\left\Vert \boldsymbol{\hat{u}}(\xi)\right\Vert _{\mathbb{C}^{d}}^{2}%
d\xi\right)  ^{1/2}<\infty,s\in\mathbb{R}%
\]
where $\boldsymbol{\hat{u}}$ is the Fourier transform of $\boldsymbol{u}$ in
the sense of distributions. If $\Omega$ $\ $is a domain in $\mathbb{R}^{d}$
then $H^{s}(\Omega,\mathbb{C}^{d})$ is the space of restrictions of
distributions $\boldsymbol{u}\in H^{s}(\mathbb{R}^{n},\mathbb{C}^{d})$ on
$\Omega$ with the norm
\[
\left\Vert \boldsymbol{u}\right\Vert _{H^{s}(\Omega,\mathbb{C}^{d})}%
=\inf_{l\boldsymbol{u}\in H^{s}(\mathbb{R}^{n},\mathbb{C}^{d})}\left\Vert
l\boldsymbol{u}\right\Vert _{H^{s}(\mathbb{R}^{n},\mathbb{C}^{d})}%
\]
where $l\boldsymbol{u}$ is an extension of $\boldsymbol{u}$ on $\mathbb{R}%
^{d}.$ If $s>1/2$ then the distributions in $H^{s}(\Omega,\mathbb{C}^{d})$
have the traces on $\partial\Omega,$ and we denote by $H^{s-1/2}%
(\partial\Omega,\mathbb{C}^{d})$ the Sobolev space on $\partial\Omega$
consisting of these traces.

\item We denote $C_{b}(\mathbb{R}^{n})$ the $C^{\ast}-$ algebra of bounded
continuous functions on $\mathbb{R}^{n},C_{b}^{k}(\mathbb{R}^{n}%
),k\in\mathbb{N}$ is a sub-algebra of $C_{b}(\mathbb{R}^{n})$ consisting of
$k-$times differentiable functions $u$ such that $\partial^{\alpha}u\in
C_{b}(\mathbb{R}^{n})$ for every multiindex $\alpha:\left\vert \alpha
\right\vert \leq k.$ If $\Omega$ is an unbounded domain in $\mathbb{R}^{n},$
then $C_{b}(\Omega),C_{b}^{k}(\Omega),C_{b}(\partial\Omega),C_{b}^{k}%
(\partial\Omega)$ are algebras consisting of the restrictions of functions in
$C_{b}(\mathbb{R}^{n}),C_{b}^{k}(\mathbb{R}^{n})$ on $\Omega,\partial\Omega,$respectively.

\item We say that a domain $\Omega\subset\mathbb{R}^{n}$ has a $C^{2}%
-$\textit{uniformly regular boundary} $\partial\Omega$ \ if $\partial\Omega$
is a $C^{2}-$ hypersurface, and: $(i)$ for fixed $r>0$ and for every point
$x_{0}\in$ $\partial\Omega$ there exists a ball $B_{r}(x_{0})=\left\{
x\in\mathbb{R}^{3}:\left\vert x-x_{0}\right\vert <r\right\}  $ and the
homeomorphism $\varphi_{x_{0}}:B_{r}(x_{0})\rightarrow B_{1}(0)$ such that
\begin{align*}
\varphi_{x_{0}}\left(  B_{r}(x_{0})\cap\Omega\right)   &  =B_{1}%
(0)\cap\mathbb{R}_{+}^{n},\mathbb{R}_{+}^{n}=\left\{  y\in\mathbb{R}^{n}%
:y_{n}>0\right\}  ,\\
\varphi_{x_{0}}\left(  B_{r}(x_{0})\cap\partial\Omega\right)   &
=B_{1}(0)\cap\mathbb{R}^{n-1},\mathbb{R}^{n-1}=\left\{  y\in\mathbb{R}%
^{n}:y_{n}=0\right\}  ;
\end{align*}
\ \ \ \ $(ii)$ if $\varphi_{x_{0}}^{i},i=1,...,n$ are the coordinate functions
of the mappings $\varphi_{x_{0}}$ then%
\[
\sup_{x_{0}\in\partial\Omega}\sup_{\left\vert \alpha\right\vert \leq2,x\in
B_{r}(x_{0})}\left\vert \partial_{x}^{\alpha}\varphi_{x_{0}}^{i}(x)\right\vert
<\infty,i=1,...,n.
\]

\item Note if $\Omega$ is a bounded domain with $C^{2}-$boundary
$\partial\Omega$, then $\partial\Omega$ is the uniformly regular surface.
\end{itemize}

\subsection{Free Dirac operators}

$\bigskip$We denote by
\[
\mathfrak{D}=\boldsymbol{\alpha}\cdot i\boldsymbol{\nabla}+\mathfrak{\alpha
}_{0}m=%
{\displaystyle\sum\limits_{j=1}^{3}}
\mathfrak{\alpha}_{j}D_{x_{j}}+\mathfrak{\alpha}_{0}m,D_{x_{j}}=i\partial
_{x_{j}}%
\]
the free Dirac operator (see for instance \cite{Thaller}) where
$\mathfrak{\alpha}_{j},j=0,1,2,3$ are the $4\times4$ Dirac matrices
\begin{equation}
\mathfrak{\alpha}_{0}=\left(
\begin{array}
[c]{cc}%
I_{2} & 0\\
0 & -I_{2}%
\end{array}
\right)  ,\alpha_{j}=\left(
\begin{array}
[c]{cc}%
0 & \sigma_{j}\\
\sigma_{j} & 0
\end{array}
\right)  ,j=1,2,3, \label{1.1}%
\end{equation}%
\begin{equation}
\mathfrak{\sigma}_{1}=\left(
\begin{array}
[c]{cc}%
0 & 1\\
1 & 0
\end{array}
\right)  ,\sigma_{2}=\left(
\begin{array}
[c]{cc}%
0 & -i\\
i & 0
\end{array}
\right)  ,\sigma_{3}=\left(
\begin{array}
[c]{cc}%
1 & 0\\
0 & -1
\end{array}
\right)  \label{1.2}%
\end{equation}
are the $2\times2$ Pauli matrices satisfying the relations
\begin{equation}
\sigma_{j}\sigma_{k}+\sigma_{k}\sigma_{j}=2\delta_{jk}I_{2},j,k=1,2,3.
\label{1.3}%
\end{equation}
Relations (\ref{1.3}) yield that
\begin{equation}
\mathfrak{\alpha}_{j}\mathfrak{\alpha}_{k}\mathfrak{+\alpha}_{k}%
\mathfrak{\alpha}_{j}=2\delta_{jk}I_{4};j,k=0,1,2,3 \label{1.4}%
\end{equation}
where $I_{n}$ is the $n\times n$ unit matrix. Equality (\ref{1.4}) implies
that
\[
\left(  \boldsymbol{\alpha}\cdot i\boldsymbol{\nabla}\right)  ^{2}=-\Delta
I_{4}%
\]
where $\Delta=\partial_{x_{1}}^{2}+\partial_{x_{2}}^{2}+\partial_{x_{3}}^{2}$
is the $3D-$Laplacian. Moreover,
\[
\mathfrak{D}^{2}=\left(  -\Delta+m^{2}\right)  I_{4}.
\]
It is well-known that the unbounded in $L^{2}(\mathbb{R}^{3},\mathbb{C}^{4})$
operator $\mathcal{D}$ generated by the free Dirac operator $\mathfrak{D}$
with domain $H^{1}(\mathbb{R}^{3},\mathbb{C}^{4})$ is self-adjoint and
\[
sp\mathcal{D}=sp_{ess}\mathcal{D}=\left(  -\infty,-\left\vert m\right\vert
\right]
{\displaystyle\bigcup}
\left[  \left\vert m\right\vert ,+\infty\right)
\]
(see for instance \cite{Thaller}).

\section{Parameter-dependent boundary problems for $3-D$ Dirac operators}

\subsection{Lopatinsky-Shapiro conditions}

We consider the parameter-dependent operator of boundary value problem%

\begin{equation}
\mathbb{D}_{\boldsymbol{A,\Phi,}\mathfrak{B}}(i\mu)\boldsymbol{u=}\left\{
\begin{array}
[c]{c}%
(\mathfrak{D}_{\boldsymbol{A},\Phi}-i\mu I_{4})\boldsymbol{u}=\mathfrak{D}%
_{\boldsymbol{A},\Phi}(i\mu)\boldsymbol{u=f}\text{\ \ on }\Omega,\mu
\in\mathbb{R},\\
\mathfrak{B}\boldsymbol{u}=\mathfrak{b}_{1}\boldsymbol{u}_{\partial\Omega}%
^{1}+\mathfrak{b}_{2}\boldsymbol{u}_{\partial\Omega}^{2}=\boldsymbol{0}\text{
\ on }\partial\Omega
\end{array}
\right.  , \label{2.1}%
\end{equation}
where $\Omega\subset\mathbb{R}^{3}$ is a domain with $C^{2}-$uniformly regular
boundary, and
\begin{equation}
\boldsymbol{A}\in L^{\infty}(\Omega,\mathbb{C}^{4}),\Phi\in L^{\infty}%
(\Omega),\mathfrak{b}_{k}=\left(  b_{k}^{ij}\right)  _{i,j=1}^{2},b_{k}%
^{ij}\in C_{b}(\partial\Omega),k=1,2, \label{2.2}%
\end{equation}
$\boldsymbol{\ u}=(\boldsymbol{u}^{1},\boldsymbol{u}^{2})\in H^{1}%
(\Omega,\mathbb{C}^{4})=H^{1}(\Omega,\mathbb{C}^{2})\oplus H^{1}%
(\Omega,\mathbb{C}^{2}),\boldsymbol{u}_{\partial\Omega}^{j}\in H^{1/2}%
(\partial\Omega,\mathbb{C}^{2})$ are boundary values of the vector-functions
$\boldsymbol{u}^{j}$ on $\partial\Omega,j=1,2.$ We study the invertibility of
the operator $\mathbb{D}_{\boldsymbol{A,\Phi,}\mathfrak{B}}(i\mu):H^{1}%
(\Omega,\mathbb{C}^{4})\rightarrow L^{2}(\Omega,\mathbb{C}^{4}),\mu
\in\mathbb{R}$ for large values of $\left\vert \mu\right\vert .$ We follows
the well-known paper \cite{AgrVishik} where the authors consider the
invertibility of general parameter-dependent boundary value problems in smooth
bounded domains.

We construct the locally inverse operators at every fixed point $x\in
\bar{\Omega}$ and then we obtain the globally inverse operator for large
values of $\left\vert \mu\right\vert $ by gluing together the locally inverse
operators by means of a countable partition of unity of finite multiplicity.

We need the Sobolev spaces $H_{\mu}^{s}(\mathbb{R}^{3},\mathbb{C}^{4})$ of
distributions $\boldsymbol{u}\in H^{s}(\mathbb{R}^{3},\mathbb{C}^{4})$ with
the norm depending on the parameter $\mu\in\mathbb{R}$
\[
\left\Vert \boldsymbol{u}\right\Vert _{\left\Vert u\right\Vert _{H_{\mu}%
^{s}(\mathbb{R}^{n},\mathbb{C}^{4})}}=\left(  \int_{\mathbb{R}^{3}%
}(1+\left\vert \xi\right\vert ^{2}+\mu^{2})^{s}\left\vert \boldsymbol{\hat{u}%
}(\xi)\right\vert ^{2}d\xi\right)  ^{1/2}<\infty,s\geq0,\mu\in\mathbb{R}%
\]
where $\boldsymbol{\hat{u}}(\xi)=\int_{\mathbb{R}^{n}}e^{ix\cdot\xi
}\boldsymbol{u}(x)dx$ is the Fourier transform. We denote by $H_{\mu}%
^{s}(\Omega,\mathbb{C}^{4})$ the space of the restrictions on the domain
$\Omega$ the distributions $\boldsymbol{u}\in H_{\mu}^{s}(\mathbb{R}%
^{3},\mathbb{C}^{4})$ with the norm
\[
\left\Vert \boldsymbol{u}\right\Vert _{H_{\mu}^{s}(\Omega,\mathbb{C}^{4}%
)}=\inf_{l\boldsymbol{u}\in H_{\mu}^{s}(\mathbb{R}^{3},\mathbb{C}^{4}%
)}\left\Vert l\boldsymbol{u}\right\Vert _{H_{\mu}^{s}(\mathbb{R}%
^{3},\mathbb{C}^{4})}%
\]
where $l\boldsymbol{u}$ is the extension of $\boldsymbol{u}\in$ $H_{\mu}%
^{s}(\Omega,\mathbb{C}^{4})$ on $\mathbb{R}^{3}$. Note that
\begin{equation}
\left\Vert \boldsymbol{u}\right\Vert _{H_{\mu}^{s}\Omega,\mathbb{C}^{4})}%
\leq(1+\mu^{2})^{-\frac{(r-s)}{2}}\left\Vert \boldsymbol{u}\right\Vert
_{H_{\mu}^{r}(\Omega,\mathbb{C}^{4})},r\geq s,\mu\in\mathbb{R} \label{2.3}%
\end{equation}
and the trace operator $\gamma_{\partial\Omega}$ is bounded from $H_{\mu}%
^{1}(\Omega,\mathbb{C}^{4})$ into $H_{\mu}^{1/2}(\partial\Omega,\mathbb{C}%
^{4})$ if $\partial\Omega$ is a $C^{l}-$surface, $l>s.$ We consider the
operator $\mathbb{D}_{\boldsymbol{A,}\Phi,\mathfrak{B}}$ as acting from
$H_{\mu}^{1}(\Omega,\mathbb{C}^{4})$ into $L^{2}(\Omega,\mathbb{C}^{4}).$

Let $\mathfrak{D}(\mu)=\boldsymbol{\alpha\cdot}i\boldsymbol{\nabla}-i\mu
I_{4}$ be the main part of the parameter-dependent operator $\mathfrak{D}%
_{\boldsymbol{A},\Phi}(i\mu)=\mathfrak{D}_{\boldsymbol{A},\Phi}-i\mu I_{4}.$
Since
\[
(\boldsymbol{\alpha\cdot\xi}+i\mu I_{4})\left(  \boldsymbol{\alpha\cdot\xi
}-i\mu I_{4}\right)  =\left(  \left\vert \boldsymbol{\xi}\right\vert ^{2}%
+\mu^{2}\right)  I_{4}%
\]
the operator $\mathfrak{D}_{\boldsymbol{A},\Phi}-i\mu I_{4}$ is the uniformly
elliptic operator with parameter $\mu\in\mathbb{R}$ \ (see for instance
\cite{AgrVishik},\cite{Agran}).

For a point $x_{0}\in\partial\Omega$ we fix the local system of orthogonal
coordinates $y=(y_{1},y_{2},y_{3})$ where the $y^{\prime}=(y_{1},y_{2})$
belongs to the tangent plane to $\partial\Omega$ at the point $x_{0}$ and the
axis $y_{3}=z$ is directed along the inward normal vector $\boldsymbol{\nu
}_{x_{0}}$\ \ to $\partial\Omega$ at the point $x_{0},$ and we associate with
the point $x_{0}\in\partial\Omega$ the family of $1-D$ boundary value problems
on the half-line $\mathbb{R}_{+}=\left\{  z\in\mathbb{R}:z>0\right\}  $
\begin{align}
&  \mathbb{\hat{D}}_{\mathfrak{B}_{x_{0}}}(\boldsymbol{\xi}^{\prime}%
,\mu)\boldsymbol{\psi}(z)\label{2.5}\\
&  =\left\{
\begin{array}
[c]{c}%
\left(  \boldsymbol{\alpha}^{\prime}\cdot\boldsymbol{\xi}^{\prime}+i\alpha
_{3}\frac{d}{dz}-i\mu I_{4}\right)  \boldsymbol{\psi}(z),z\in\mathbb{R}%
_{+},\boldsymbol{\alpha}^{\prime}\cdot\boldsymbol{\xi}^{\prime}=\alpha_{1}%
\xi_{1}+\alpha_{2}\xi_{2}\\
\mathfrak{B}_{x_{0}}\boldsymbol{\psi}(0)=\mathfrak{b}^{1}(x_{0}%
)\boldsymbol{\psi}^{1}(+0)+\mathfrak{b}^{2}(x_{0})\boldsymbol{\psi}%
^{2}(\boldsymbol{+}0)\boldsymbol{=0}%
\end{array}
\right. \nonumber
\end{align}
acting from $H^{1}(\mathbb{R}_{+},\mathbb{C}^{4})$ into $L^{2}(\mathbb{R}%
_{+},\mathbb{C}^{4}),$where $\boldsymbol{\psi}^{j}(+0)=\lim_{z\rightarrow
+0}\boldsymbol{\psi}^{j}(z),j=1,2.$

We are looking for the exponentially decreasing solutions of the equation
\begin{equation}
\mathbb{\hat{D}}_{\mathfrak{B}_{x_{0}}}(\boldsymbol{\xi}^{\prime}%
,\mu)\boldsymbol{\psi=0}\text{ on }\mathbb{R}_{+}. \label{2.6}%
\end{equation}
Note that
\begin{equation}
\left(  \boldsymbol{\alpha}^{\prime}\cdot\boldsymbol{\xi}^{\prime
}+i\mathfrak{\alpha}_{3}\frac{d}{dz}+i\mu I_{4}\right)  \left(
\boldsymbol{\alpha}^{\prime}\cdot\boldsymbol{\xi}^{\prime}+i\mathfrak{\alpha
}_{3}\frac{d}{dz}-i\mu I_{4}\right)  =\left(  \left\vert \boldsymbol{\xi
}^{\prime}\right\vert ^{2}+\mu^{2}-\frac{d^{2}}{dz^{2}}\right)  I_{4}.
\label{2.7}%
\end{equation}
Hence the equation%
\begin{equation}
\left(  \boldsymbol{\alpha}^{\prime}\cdot\boldsymbol{\xi}^{\prime
}+i\mathfrak{\alpha}_{3}\frac{d}{dz}-i\mu I_{4}\right)  \boldsymbol{\psi}(z)=0
\label{2.8}%
\end{equation}
has the exponentially decreasing solutions on $\mathbb{R}_{+}$ of the form
\begin{equation}
\boldsymbol{\psi}(z)=\boldsymbol{h}\left(  \boldsymbol{\xi}^{\prime}%
,\mu\right)  e^{-\rho z},\rho=\sqrt{\left\vert \boldsymbol{\xi}^{\prime
}\right\vert ^{2}+\mu^{2}}>0,z>0,(\boldsymbol{\xi}^{\prime},\mu)\in
\mathbb{R}^{3} \label{2.9}%
\end{equation}
where the vector $\boldsymbol{h=h}\left(  \boldsymbol{\xi}^{\prime}%
,\mu\right)  \in\mathbb{C}^{4}$ satisfies the equation
\begin{equation}
\left(  \boldsymbol{\alpha}^{\prime}\cdot\boldsymbol{\xi}^{\prime}-i\alpha
_{3}\rho-i\mu I_{4}\right)  \boldsymbol{h}=0. \label{2.10}%
\end{equation}
The general solution of equation (\ref{2.10}) has the form
\begin{equation}
\boldsymbol{h}\left(  \boldsymbol{\xi}^{\prime},\mu\right)  =\Theta\left(
\boldsymbol{\xi}^{\prime},\mu\right)  \boldsymbol{f} \label{2.11}%
\end{equation}
where $\boldsymbol{f}$ $\in\mathbb{C}^{4}$ is an arbitrary vector, and
\[
\Theta\left(  \boldsymbol{\xi}^{\prime},\mu\right)  =\alpha^{\prime}%
\cdot\boldsymbol{\xi}^{\prime}+i\rho\alpha_{3}+i\mu I_{4}=\left(
\begin{array}
[c]{cc}%
i\mu I_{2} & \Lambda(\boldsymbol{\xi}^{\prime},\mu)\\
\Lambda(\boldsymbol{\xi}^{\prime},\mu) & i\mu I_{2}%
\end{array}
\right)  ,
\]
\begin{align}
\Lambda &  =\Lambda(\boldsymbol{\xi}^{\prime},\mu)=\sigma^{\prime}%
\cdot\boldsymbol{\xi}^{\prime}-i\rho\sigma_{3}=\left(
\begin{array}
[c]{cc}%
-i\rho & \bar{\varsigma}\\
\varsigma & i\rho
\end{array}
\right)  ,\varsigma=\xi_{1}+i\xi_{2},\label{2.12}\\
\sigma^{\prime}\cdot\boldsymbol{\xi}^{\prime}  &  =\sigma_{1}\xi_{1}%
+\sigma_{2}\xi_{2}.\nonumber
\end{align}

We set%

\[
\boldsymbol{e}_{1}=\left(
\begin{array}
[c]{c}%
1\\
0
\end{array}
\right)  ,\boldsymbol{e}_{2}=\left(
\begin{array}
[c]{c}%
0\\
1
\end{array}
\right)  ,\boldsymbol{0}=\left(
\begin{array}
[c]{c}%
0\\
0
\end{array}
\right)  ,
\]
and
\[
\boldsymbol{h}_{1}=\boldsymbol{h}_{1}(\boldsymbol{\xi}^{\prime},\mu
)=\Theta(\boldsymbol{\xi}^{\prime},\mu)\left(
\begin{array}
[c]{c}%
\boldsymbol{e}_{1}\\
\boldsymbol{0}%
\end{array}
\right)  =\left(
\begin{array}
[c]{c}%
i\mu\boldsymbol{e}_{1}\\
\Lambda(\boldsymbol{\xi}^{\prime},\mu)\boldsymbol{e}_{1}%
\end{array}
\right)  ,
\]%
\begin{equation}
\boldsymbol{h}_{2}=\boldsymbol{h}_{2}(\boldsymbol{\xi}^{\prime},\mu
)=\Theta(\boldsymbol{\xi}^{\prime},\mu)\left(
\begin{array}
[c]{c}%
0\\
\boldsymbol{e}_{2}%
\end{array}
\right)  =\left(
\begin{array}
[c]{c}%
\Lambda(\boldsymbol{\xi}^{\prime},\mu)\boldsymbol{e}_{2}\\
i\mu\boldsymbol{e}_{2}%
\end{array}
\right)  \label{2.13}%
\end{equation}
where
\[
\Lambda(\boldsymbol{\xi}^{\prime},\mu)\boldsymbol{e}_{1}=\left(
\begin{array}
[c]{c}%
-i\rho\\
\varsigma
\end{array}
\right)  ,\Lambda(\boldsymbol{\xi}^{\prime},\mu)\boldsymbol{e}_{2}=\left(
\begin{array}
[c]{c}%
\bar{\varsigma}\\
i\rho
\end{array}
\right)  .
\]
The vectors $\boldsymbol{h}_{1}(\xi^{\prime},\mu),\boldsymbol{h}_{2}%
(\xi^{\prime},\mu)$ are orthogonal and satisfy equation (\ref{2.10}). Hence
\[
\left\{  \boldsymbol{h}_{1}(\xi^{\prime},\mu)e^{-\rho z},\boldsymbol{h}%
_{2}(\xi^{\prime},\mu)e^{-\rho z}\right\}
\]
is a fundamental system of solutions of equation (\ref{2.8}) in $L^{2}%
(\mathbb{R}_{+},\mathbb{C}^{4}).$ Then every solution $\boldsymbol{u}\in
L^{2}(\mathbb{R}_{+},\mathbb{C}^{4})$ of equation (\ref{2.8}) on
$\mathbb{R}_{+}$ is of the form%
\begin{equation}
\boldsymbol{\psi}(z)=C_{1}\boldsymbol{h}_{1}e^{-\rho z}+C_{2}\boldsymbol{h}%
_{2}e^{-\rho z},\rho>0 \label{2.14}%
\end{equation}
where $C_{1},C_{2}\in\mathbb{C}$. Substituting $\boldsymbol{\psi}$ \ in the
boundary condition
\[
\mathfrak{b}_{1}(x_{0})\boldsymbol{\psi}^{1}(+0)+\mathfrak{b}(x_{0}%
)\boldsymbol{\psi}^{2}(+0)=\boldsymbol{0}%
\]
we\ obtain a system of linear equations
\begin{equation}
\left(  \mathfrak{b}_{1}(x_{0})\boldsymbol{h}_{1}^{1}+\mathfrak{b}_{2}%
(x_{0})\boldsymbol{h}_{1}^{2}\right)  C_{1}+(\mathfrak{b}_{1}(x_{0}%
)\boldsymbol{h}_{2}^{1}+\mathfrak{b}_{2}(x_{0})\boldsymbol{h}_{2}^{2}%
)C_{2}=\boldsymbol{0} \label{2.15}%
\end{equation}
with respect to $C_{1},C_{2}.$ Let $\mathcal{L(}x_{0},\boldsymbol{\xi}%
^{\prime},\mu)=(b_{1}(x_{0},\boldsymbol{\xi}^{\prime},\mu),b_{2}%
(x_{0},\boldsymbol{\xi}^{\prime},\mu))$ be the matrix with columns
\begin{align}
b_{1}(x_{0},\boldsymbol{\xi}^{\prime},\mu)  &  =\mathfrak{b}_{1}%
(x_{0})\boldsymbol{h}_{1}^{1}\left(  \boldsymbol{\xi}^{\prime},\mu\right)
+\mathfrak{b}_{2}(x_{0})\boldsymbol{h}_{1}^{2}\left(  \boldsymbol{\xi}%
^{\prime},\mu\right)  ,\label{2.16}\\
b_{2}(x_{0},\boldsymbol{\xi}^{\prime},\mu)  &  =\mathfrak{b}_{1}%
(x_{0})\boldsymbol{h}_{2}^{1}\left(  \boldsymbol{\xi}^{\prime},\mu\right)
+\mathfrak{b}_{2}(x_{0})\boldsymbol{h}_{2}^{2}\left(  \boldsymbol{\xi}%
^{\prime},\mu\right)  .\nonumber
\end{align}
System (\ref{2.15}) has the trivial solution if and only if
\[
\det\mathcal{L(}x_{0},\boldsymbol{\xi}^{\prime},\mu)\neq0
\]

\begin{definition}
\label{de2.1} $(i)$ We say that the operator $\mathbb{D}_{\boldsymbol{A,\Phi
,}\mathfrak{B}}(i\mu)$ satisfies the local parameter-dependent
Lopatinsky-Shapiro condition at \ $x_{0}\in\partial\Omega$ if
\begin{equation}
\det\mathcal{L(}x_{0},\boldsymbol{\xi}^{\prime},\mu)\neq0\text{ for every
}\left(  \boldsymbol{\xi}^{\prime},\mu\right)  :\left\vert \boldsymbol{\xi
}^{\prime}\right\vert ^{2}+\mu^{2}=1. \label{2.17}%
\end{equation}
$(ii)$ We say that the operator $\mathbb{D}_{\boldsymbol{A,\Phi,}\mathfrak{B}%
}(i\mu)$ satisfies the uniform parameter-dependent Lopatinsky-Shapiro
condition if
\begin{equation}
\inf_{x\in\partial\Omega,\mu^{2}+\left\vert \boldsymbol{\xi}^{\prime
}\right\vert ^{2}=1}\left\vert \det\mathcal{L(}x,\boldsymbol{\xi}^{\prime}%
,\mu)\right\vert >0. \label{2.18}%
\end{equation}

\end{definition}

Note that if the boundary $\partial\Omega$ is a compact set and the local
parameter-dependent Lopatinsky-Shapiro condition (\ref{2.17}) is satisfied at
every point $x\in\partial\Omega,$ then the uniform Lopatinsky-Shapiro
condition (\ref{2.18}) holds.

\subsection{Standard Lopatinsky-Shapiro condition}

We introduce the standard Lopatinsky-Shapiro condition for the operator
$\mathbb{D}_{\boldsymbol{A,}\Phi,\mathfrak{B}}$ for domains $\Omega
\subset\mathbb{R}^{3}$ with a $C^{2}$-boundary. As above for a fix point
$x_{0}\in\partial\Omega$ we introduce the local system of coordinates
$y=(y_{1},y_{2},y_{3})$ where the axis $y_{1},y_{2}\in\mathbb{T}_{x_{0}}$ the
tangent plane to $\partial\Omega$ at the point $x_{0}$ and the axis $y_{3}=z$
is directed along the outward normal vector $\boldsymbol{\nu}$ \ to
$\partial\Omega$ at the point $x_{0}.$

Let
\begin{equation}
\boldsymbol{\psi}_{1}(\xi^{\prime},z)=\boldsymbol{h}_{1}(\xi^{\prime
})e^{-\left\vert \xi^{\prime}\right\vert z},\boldsymbol{\psi}_{2}(\xi^{\prime
},z)=\boldsymbol{h}_{2}(\xi^{\prime})e^{-\left\vert \xi^{\prime}\right\vert
z},z>0
\end{equation}
where
\begin{equation}
\boldsymbol{h}_{1}(\xi^{\prime})=\left(
\begin{array}
[c]{c}%
\Lambda(\xi^{\prime})\boldsymbol{e}_{1}\\
\boldsymbol{0}%
\end{array}
\right)  ,\boldsymbol{h}_{2}(\xi^{\prime})=\left(
\begin{array}
[c]{c}%
\boldsymbol{0}\\
\Lambda(\xi^{\prime})\boldsymbol{e}_{2}%
\end{array}
\right)  .
\end{equation}
We introduce the $2\times2$ matrix
\[
\mathcal{L}(x_{0},\xi^{\prime})=(\mathfrak{b}_{1}(x_{0})\Lambda(\xi^{\prime
})\boldsymbol{e}_{1},\mathfrak{b}_{2}(x_{0})\Lambda(\xi^{\prime}%
)\boldsymbol{e}_{2})
\]
which coincides with the matrix $\mathcal{L}(x_{0},\xi^{\prime},0)$.

\begin{definition}
\label{d2.2}We say that the operator $\mathbb{D}_{\boldsymbol{A,}%
\Phi,\mathfrak{B}}$ satisfies the local standard Lopatinsky-Shapiro condition
at the point $x_{0}\in\partial\Omega$ if
\begin{equation}
\det\mathcal{L}(x_{0},\xi^{\prime})\neq0\text{ for all }\xi^{\prime}\in S^{1}
\label{2.18'}%
\end{equation}
and the operator $\mathbb{D}_{\boldsymbol{A,}\Phi,\mathfrak{B}}$ satisfies the
standard Lopatinsky-Shapiro conditions \ uniformly if
\begin{equation}
\inf_{x_{0}\in\partial\Omega,\xi^{\prime}\in S^{1}}\left\vert \det
\mathcal{L}(x_{0},\xi^{\prime})\right\vert >0. \label{2.18''}%
\end{equation}

\end{definition}

\subsection{Invertibility of $\mathbb{D}_{\boldsymbol{A,}\Phi,\mathfrak{B}%
}(i\mu)$ for large values of $\left\vert \mu\right\vert $}

\begin{theorem}
\label{t2.2} Let $\Omega\subset\mathbb{R}^{3}$ be a domain with a $C^{2}%
-$uniformly regular boundary $\partial\Omega$, the magnetic potential
$\boldsymbol{A=}(A_{1},A_{2},A_{3})\in L^{\infty}(\Omega,\mathbb{C}^{3}),$ the
electrostatic potential $\Phi\in L^{\infty}(\Omega),$ $\mathfrak{b}_{j}\in
C_{b}(\partial\Omega)\otimes\mathcal{B(}\mathbb{C}^{2}),j=1,2,$ and the
uniform parameter-dependent Lopatinsky condition (\ref{2.18}) for
$\mathbb{D}_{\boldsymbol{A},\Phi,\mathfrak{B}}(i\mu),\mu\in\mathbb{R}$ be
satisfied. Then there exists $\mu_{0}>0$ such that the operator $\mathbb{D}%
_{\boldsymbol{A,\Phi,}\mathfrak{B}}(i\mu):H^{1}(\Omega,\mathbb{C}%
^{4})\rightarrow L^{2}(\Omega,\mathbb{C}^{4})$ is invertible for every $\mu
\in\mathbb{R}:\left\vert \mu\right\vert >\mu_{0}.$
\end{theorem}

\begin{proof}
The idea of the proof is similar to given in the paper \cite{AgrVishik} (see
also \cite{Agran}, Sec.3). However, we consider the parameter-dependent
boundary problems for unbounded domains, and therefore we need an infinite
partition of unity and estimates associated with them. \ 

Since the Dirac operator $\mathfrak{D}_{\boldsymbol{A},\Phi}(i\mu)$ is the
uniformly elliptic parameter-depending operator on $\bar{\Omega}$, and the
uniform Lopatinsky-Shapiro condition (\ref{2.18}) holds, and $\partial\Omega$
is a $C^{2}-$uniformly regular surface, then there exists $r>0$ and $\mu
_{0}>0$ such that for every point $x_{0}\in\bar{\Omega}$ there exist
operators
\[
L_{x_{0}}(\mu),R_{x_{0}}(\mu)\in\mathcal{B}(L^{2}(\Omega,\mathbb{C}%
^{4}),H_{\mu}^{1}(\Omega,\mathbb{C}^{4}))\text{ }%
\]
such that
\begin{align}
\sup_{x_{0}\in\bar{\Omega},\left\vert \mu\right\vert \geq\mu_{0}}\left\Vert
L_{x_{0}}(\mu)\right\Vert _{\mathcal{B(}L^{2}(\Omega,\mathbb{C}^{4}%
),H^{1}(\Omega,\mathbb{C}^{4}))}  &  =d_{L}<\infty,\label{2.21}\\
\sup_{x_{0}\in\bar{\Omega},\left\vert \mu\right\vert \geq\mu_{0}}\left\Vert
R_{x_{0}}(\mu)\right\Vert _{\mathcal{B(}L^{2}(\Omega,\mathbb{C}^{4}%
),H^{1}(\Omega,\mathbb{C}^{4}))}  &  =d_{R}<\infty,\nonumber
\end{align}
and for every function $\varphi\in C_{0}^{\infty}(B_{r}(x_{0}))$
\begin{align}
L_{x_{0}}(\mu)\mathbb{D}_{\boldsymbol{A},\Phi,\mathfrak{B,\Omega}}%
(i\mu)\varphi I  &  =\varphi I,\nonumber\\
\varphi\mathbb{D}_{\boldsymbol{A},\Phi,\mathfrak{B,\Omega}}(i\mu)R_{x_{0}}%
(\mu)  &  =\varphi I. \label{2.22}%
\end{align}
We choose from the coverage $%
{\displaystyle\bigcup\limits_{x_{0}\in\bar{\Omega}}}
B_{r}(x_{0})\supset\bar{\Omega}$ a countable sub-coverage $%
{\displaystyle\bigcup\limits_{j\in\mathbb{N}}}
B_{r}(x_{j})\supset\bar{\Omega}$ of a finite multiplicity $N\in\mathbb{N},$
and we construct the partition of unity
\begin{equation}
\sum_{j\in\mathbb{N}}\theta_{j}(x)=1,x\in\bar{\Omega} \label{2.22'}%
\end{equation}
subordinated to the coverage $%
{\displaystyle\bigcup\limits_{j\in\mathbb{N}}}
B_{r}(x_{j})$, with $\theta_{j}\in C_{0}^{\infty}(B_{r}(x_{j}))$, $0\leq
\theta_{j}(x)\leq1,$ such that the sum $\sum_{j\in\mathbb{N}}\theta_{j}(x)$
contains for every $x\in\bar{\Omega}$ not more than $N$ nonzero terms. \ Let
$\varphi_{j}\in C_{0}^{\infty}(B_{r}(x_{j}),0\leq\varphi_{j}(x)\leq1,$ and
$\theta_{j}\varphi_{j}=\theta_{j}.$ We set
\begin{equation}
L(\mu)\boldsymbol{f}=\sum_{j\in\mathbb{N}}\theta_{j}L_{x_{j}}(\mu)\varphi
_{j}\boldsymbol{f,f}\in C_{0}^{\infty}(\bar{\Omega},\mathbb{C}^{4}),
\label{2.24'}%
\end{equation}%
\begin{equation}
R(\mu)\boldsymbol{f}=\sum_{j\in\mathbb{N}}\varphi_{j}R_{x_{j}}(\mu)\theta
_{j}\boldsymbol{f},\boldsymbol{f}\in C_{0}^{\infty}(\bar{\Omega}%
,\mathbb{C}^{4}) \label{2.24}%
\end{equation}
where $C_{0}^{\infty}(\bar{\Omega},\mathbb{C}^{4})$ is the space of
restrictions of functions in $C_{0}^{\infty}(\mathbb{R}^{3},\mathbb{C}^{4})$
on $\bar{\Omega}.$ Taking into account that the coverage $\left\{  B_{r}%
(x_{j})\right\}  _{j\in\mathbb{N}}$ has the finite multiplicity $N$ we obtain
the estimates
\begin{equation}
\left\Vert L(\mu)\boldsymbol{f}\right\Vert _{H_{\mu}^{1}(\Omega,\mathbb{C}%
^{4})}\leq C\sup_{j\in\mathbb{N}}\left\Vert L_{j}(\mu)\right\Vert \left\Vert
\boldsymbol{f}\right\Vert _{L^{2}(\Omega,\mathbb{C}^{4})}\leq Cd_{L}\left\Vert
\boldsymbol{f}\right\Vert _{L^{2}(\Omega,\mathbb{C}^{4})}, \label{2.25}%
\end{equation}%
\begin{equation}
\left\Vert R(\mu)\boldsymbol{f}\right\Vert _{H_{\mu}^{1}(\Omega,\mathbb{C}%
^{4})}\leq C\sup_{j\in\mathbb{N}}\left\Vert R_{j}(\mu)\right\Vert \left\Vert
f\right\Vert _{L^{2}(\Omega,\mathbb{C}^{4})}\leq Cd_{R}\left\Vert
\boldsymbol{f}\right\Vert _{L^{2}(\Omega,\mathbb{C}^{4})} \label{2.26}%
\end{equation}
for every $\left\vert \mu\right\vert \geq\mu_{0}>0,\boldsymbol{f}\in
C_{0}^{\infty}(\bar{\Omega},\mathbb{C}^{4})$ with the constant $C>0$
independent of $\boldsymbol{f.}$ \ Estimates (\ref{2.25}), (\ref{2.26}) yield
that the operators $L(\mu),R(\mu)$ are continued to \ bounded operators acting
from $L^{2}(\Omega,\mathbb{C}^{4})$ into $H_{\mu}^{1}(\Omega,\mathbb{C}%
^{4})).$ Let $\psi_{j}\in C_{0}^{\infty}(B_{r}(x_{j})),0\leq\psi_{j}(x)\leq1$
$\varphi_{j}\psi_{j}=\varphi_{j},\psi_{j}\in C_{0}^{\infty}(B_{r}(x_{j})),$
$0\leq\psi_{j}(x)\leq1$ be such that $\varphi_{j}\psi_{j}=\varphi_{j}.$ Then
\begin{align}
L(\mu)\mathbb{D}_{\boldsymbol{A},\Phi,\mathfrak{B}}(i\mu)  &  =\sum
_{j\in\mathbb{N}}\theta_{j}L_{x_{j}}(\mu)\varphi_{j}\mathbb{D}_{\boldsymbol{A}%
,\Phi,\mathfrak{B}}(i\mu)\psi_{j}I=\label{2.27}\\
&  I+T_{1}(\mu),\nonumber
\end{align}
where
\[
T_{1}(\mu)=\sum_{j\in\mathbb{N}}\theta_{j}L_{x_{j}}(\mu)\left[  \mathbb{D}%
_{\boldsymbol{A},\Phi,\mathfrak{B,\Omega}}(i\mu),\varphi_{j}I\right]  \psi
_{j},
\]
and%
\[
\left[  \mathbb{D}_{\boldsymbol{A},\Phi,\mathfrak{B}}(i\mu),\varphi
_{j}I\right]  =\mathbb{D}_{\boldsymbol{A},\Phi,\mathfrak{B}}(\mu)\varphi
_{j}I-\varphi_{j}\mathbb{D}_{\boldsymbol{A},\Phi,\mathfrak{B}}(\mu).
\]
Applying estimate (\ref{2.3}) we obtain that
\begin{equation}
\left\Vert \left[  \mathbb{D}_{\boldsymbol{A},\Phi,\mathfrak{B}}(i\mu
),\varphi_{j}I\right]  \right\Vert _{\mathcal{B}(H_{\mu}^{1}(\Omega
,\mathbb{C}^{4}),L^{2}(\Omega,\mathbb{C}^{4}))}\leq\frac{C}{\left\vert
\mu\right\vert },\left\vert \mu\right\vert \geq\mu_{0} \label{2.28}%
\end{equation}
with the constant $C>0$ independent of $j\in\mathbb{N}$. Again applying the
finite multiplicity of the coverage $\left\{  B_{r}(x_{j})\right\}
_{j\in\mathbb{N}}$ we obtain from (\ref{2.25}), (\ref{2.28}) the estimate
\begin{equation}
\left\Vert T_{1}(\mu)\right\Vert _{\mathcal{B(}H^{1}(\Omega,\mathbb{C}%
^{4}),H^{1}(\Omega,\mathbb{C}^{4}))}\leq\frac{C}{\left\vert \mu\right\vert
}\sup_{j\in\mathbb{N}}\left\Vert L_{x_{j}}(\mu)\right\Vert \leq\frac{Cd_{L}%
}{\left\vert \mu\right\vert },\left\vert \mu\right\vert \geq\mu_{0}.
\label{2.29}%
\end{equation}
Hence there exists $\mu_{1}\geq\mu_{0}$ such that
\begin{equation}
\sup_{\left\vert \mu\right\vert \geq\mu_{1}}\left\Vert T_{1}(\mu)\right\Vert
_{\mathcal{B(}H^{1}(\Omega,\mathbb{C}^{4}),H^{1}(\Omega,\mathbb{C}^{4}))}<1.
\label{2.30}%
\end{equation}
Thus the operator $\mathbb{D}_{\boldsymbol{A},\Phi,\mathfrak{B}}(i\mu)$ has
the left inverse operator $\mathbb{L(\mu)=}\left(  I+T_{1}(\mu)\right)
^{-1}L(\mu)$ for every $\mu\in\mathbb{R}:\mu\geq\mu_{1}.$ In the same way we
prove that there exists a right inverse operator $\mathbb{R(\mu)}$ of
$\mathbb{D}_{\boldsymbol{A},\Phi,\mathfrak{B,\Omega}}(i\mu)$ for $\left\vert
\mu\right\vert \geq\mu_{2}>\mu_{1}.$ Hence the operator $\mathbb{D}%
_{\boldsymbol{A},\Phi,\mathfrak{B}}(i\mu):H_{\mu}^{1}(\Omega,\mathbb{C}%
^{4})\rightarrow L^{2}(\Omega,\mathbb{C}^{4})$ is invertible for every
$\left\vert \mu\right\vert \geq\mu_{2}.$
\end{proof}

\begin{corollary}
\label{co2.1} Let conditions of Theorem \ref{t2.2} be satisfied. Then there
exists $\tilde{\mu}>0$ such that the operator $\mathbb{D}_{\boldsymbol{A}%
,\Phi,\mathfrak{B}}(i\mu):H^{1}(\Omega,\mathbb{C}^{4})\rightarrow L^{2}%
(\Omega,\mathbb{C}^{4})$ is invertible for every $\mu\in\mathbb{R}:\left\vert
\mu\right\vert \geq\tilde{\mu}.$
\end{corollary}

\begin{proof}
For every fix $\mu\in\mathbb{R}$ the norm in the space $H_{\mu}^{1}%
(\Omega,\mathbb{C}^{4})$ is equivalent to the norm in the usual Sobolev spaces
$H^{1}(\Omega,\mathbb{C}^{4})$ without parameter $\mu.$ It implies the
invertibility of the operator $\mathbb{D}_{\boldsymbol{A},\Phi,\mathfrak{B}%
}(i\mu):H^{1}(\Omega,\mathbb{C}^{4})\rightarrow L^{2}(\Omega,\mathbb{C}^{4})$
for every $\mu:\left\vert \mu\right\vert \geq\tilde{\mu}. $
\end{proof}

\section{ Self-adjointness of the unbounded operator $\mathcal{D}%
_{\boldsymbol{A},\Phi,\mathfrak{B}}$}

Now we consider the self-adjointness in $L^{2}(\mathbb{R}^{3},\mathbb{C}^{4})$
the unbounded operators $\mathcal{D}_{\boldsymbol{A},\Phi,\mathfrak{B}}$
associated with the operator $\mathbb{D}_{\boldsymbol{A,}\Phi,\mathfrak{B}}$
defined by the Dirac operator
\[
\mathfrak{D}_{\boldsymbol{A,}\Phi}=\boldsymbol{\alpha}\cdot\left(
i\boldsymbol{\nabla}+\boldsymbol{A}\right)  +\mathfrak{\alpha}_{0}m+\Phi I_{4}%
\]
where $\boldsymbol{A}\in L^{\infty}(\Omega,\mathbb{C}^{4})$, $\Phi\in
L^{\infty}(\Omega)$ with domain
\begin{equation}
H_{\mathfrak{B}}^{1}(\Omega,\mathbb{C}^{4})=\left\{
\begin{array}
[c]{c}%
\boldsymbol{u}\in H^{1}(\Omega,\mathbb{C}^{4}):\mathfrak{B}\boldsymbol{u}%
(x^{\prime})=\mathfrak{b}_{1}(x^{\prime})\boldsymbol{u}_{\partial\Omega}%
^{1}(x^{\prime})+\mathfrak{b}_{2}(x^{\prime})\boldsymbol{u}_{\partial\Omega
}^{2}(x^{\prime})=\boldsymbol{0},\\
x^{\prime}\in\partial\Omega,\mathfrak{b}_{j}\in C_{b}(\partial\Omega
,\mathcal{B}(\mathbb{R}^{2})).
\end{array}
\right\}  ,\label{3.3}%
\end{equation}

\begin{theorem}
\label{te3.1} Let: (i) $\Omega\subset\mathbb{R}^{3}$ be a domain with the
$C^{2}$-uniformly regular boundary, (ii) the vector-valued potential
$\boldsymbol{A}\in L^{\infty}(\Omega,\mathbb{R}^{4})$ and the electrostatic
potential $\Phi\in L^{\infty}(\Omega)$ be real-valued, (iii) the uniformly
parameter-dependent Lopatinsky-Shapiro condition
\[
\inf_{x^{\prime}\in\partial\Omega,\left\vert \xi^{\prime}\right\vert ^{2}%
+\mu^{2}=1}\left\vert \det\mathcal{L(}x^{\prime},\boldsymbol{\xi}^{\prime}%
,\mu)\right\vert >0
\]
hold; (iv) $\mathfrak{b}_{j}\in C_{b}(\partial\Omega)\otimes\mathcal{B(}%
\mathbb{C}^{2}),j=1,2$ and the operator $\mathcal{D}_{\boldsymbol{A}%
,\Phi,\mathfrak{B}}$ is symmetric in $L^{2}(\Omega,\mathbb{C}^{4}).$ Then the
operator $\mathcal{D}_{\boldsymbol{A},\Phi,\mathfrak{B}}$ is self-adjoint in
$L^{2}(\Omega,\mathbb{C}^{4}).$
\end{theorem}

\begin{proof}
\ Corollary \ref{co2.1} yields that there exists $\left\vert \mu\right\vert $
large enough and a constant $C>0$ such that for every $\boldsymbol{u}\in
H_{\mathfrak{B}}^{1}(\Omega,\mathbb{C}^{4})$
\begin{equation}
\left\Vert \boldsymbol{u}\right\Vert _{H^{1}(\Omega,\mathbb{C}^{4})}\leq
C\left(  \left\Vert \mathfrak{D}_{\boldsymbol{A},\Phi}\boldsymbol{u}%
\right\Vert _{L^{2}(\Omega,\mathbb{C}^{4})}+\left\vert \mu\right\vert
\left\Vert \boldsymbol{u}\right\Vert _{L^{2}(\Omega,\mathbb{C}^{4})}\right)  .
\label{3.6}%
\end{equation}
It follows from a priori estimate (\ref{3.6}) that the operator $\mathcal{D}%
_{\boldsymbol{A},\Phi,\mathfrak{B}}$ is closed. \ Moreover, Corollary
\ref{co2.1} yields that the deficiency indices of $\mathcal{D}_{\boldsymbol{A}%
,\Phi,\mathfrak{B}}$ are equal $0.$ Hence (see for instance \cite{BirSol},
page 100) the operator $\mathcal{D}_{\boldsymbol{A},\Phi,\mathfrak{B}}$ is self-adjoint.
\end{proof}

\begin{corollary}
\label{co3.1} Let
\[
H_{\mathfrak{B}}^{1}(\Omega,\mathbb{C}^{4})=\left\{  \boldsymbol{u}\in
H^{1}(\Omega,\mathbb{C}^{4}):\mathfrak{B}\boldsymbol{u}(x^{\prime
})=\boldsymbol{u}_{\partial\Omega}^{1}(x^{\prime})+\mathfrak{b}(x^{\prime
})\boldsymbol{u}_{\partial\Omega}^{2}(x^{\prime})=\boldsymbol{0},x^{\prime}%
\in\partial\Omega\right\}
\]
where $\mathfrak{b}_{j}\in C_{b}(\partial\Omega)\otimes\mathcal{B(}%
\mathbb{C}^{2}),$and
\begin{equation}
\mathfrak{b}^{\ast}\left(  \mathfrak{\sigma\cdot}\boldsymbol{\nu}\right)
+\left(  \mathfrak{\sigma\cdot}\boldsymbol{\nu}\right)  \mathfrak{b}=0\text{
on }\partial\Omega\label{3.3'}%
\end{equation}
where $\boldsymbol{\nu}$ is the outward unit normal vector to $\partial
\Omega.$ Then the operator $\mathcal{D}_{\boldsymbol{A},\Phi,\mathfrak{B}}$ is
symmetric. Hence, if conditions (i),(ii),(iii) of Theorem \ref{te3.1} hold
then the operator $\mathcal{D}_{\boldsymbol{A},\Phi,\mathfrak{B}}$ is
self-adjoint in $L^{2}(\Omega,\mathbb{C}^{4}).$
\end{corollary}

\begin{proof}
Integrating by parts we obtain
\begin{align*}
\left\langle \mathfrak{D}_{\boldsymbol{A,}\Phi}\boldsymbol{u},\boldsymbol{v}%
\right\rangle _{L^{2}(\Omega,\mathbb{C}^{4})}-\left\langle \boldsymbol{u}%
,\mathfrak{D}_{\boldsymbol{A,}\Phi}\boldsymbol{v}\right\rangle _{L^{2}%
(\Omega,\mathbb{C}^{4})}  &  =\left\langle \left(  -i\boldsymbol{\alpha}%
\cdot\boldsymbol{\nu}\right)  \boldsymbol{u}_{\partial\Omega},\boldsymbol{v}%
_{\partial\Omega}\right\rangle _{L^{2}(\partial\Omega,\mathbb{C}^{4})},\\
\boldsymbol{u,v}  &  \in H_{\mathfrak{B}}^{1}(\Omega,\mathbb{C}^{4}).
\end{align*}
Taking into account (\ref{3.3'}) we obtain that
\begin{align}
\left(  \boldsymbol{\alpha}\cdot\boldsymbol{\nu}\right)  \boldsymbol{u}%
_{\partial\Omega}\cdot\boldsymbol{v}_{\partial\Omega}  &  =-i\left(
\boldsymbol{\sigma}\cdot\boldsymbol{\nu}\right)  \boldsymbol{u}_{\partial
\Omega}^{2}\cdot\boldsymbol{v}_{\partial\Omega}^{1}-i\left(
\boldsymbol{\sigma}\cdot\boldsymbol{\nu}\right)  \boldsymbol{u}_{\partial
\Omega}^{1}\cdot\boldsymbol{v}_{\partial\Omega}^{2}\label{3.5}\\
&  =\left(  \mathfrak{b}^{\ast}i\left(  \boldsymbol{\sigma}\cdot
\boldsymbol{\nu}\right)  +i\left(  \boldsymbol{\sigma}\cdot\boldsymbol{\nu
}\right)  \mathfrak{b}\right)  \boldsymbol{u}_{\partial\Omega}^{1}%
\cdot\boldsymbol{v}_{\partial\Omega}^{2}=0.\nonumber
\end{align}
Hence $\mathcal{D}_{\boldsymbol{A},\Phi,\mathfrak{B}}$ is a symmetric operator
and by Theorem \ref{te3.1} $\mathcal{D}_{\boldsymbol{A},\Phi,\mathfrak{B}}$ is self-adjoint.
\end{proof}

\subsubsection{Self-adjointness of generalized MIT bag model}

We consider the operator of generalized \textit{MIT bag model}
\begin{equation}
\mathbb{M}_{\boldsymbol{A,}\Phi,\mathfrak{B}}\boldsymbol{u}(x)\boldsymbol{=}%
\left\{
\begin{array}
[c]{c}%
\mathfrak{D}_{\boldsymbol{A},\Phi}\boldsymbol{u}(x),x\in\Omega\\
\mathfrak{M}_{a}(x^{\prime})\boldsymbol{u}_{\partial\Omega}(x^{\prime
})=\boldsymbol{u}_{\partial\Omega}^{1}(x^{\prime})+ia(x^{\prime})\left(
\sigma\cdot\boldsymbol{\nu}\right)  \boldsymbol{u}_{\partial\Omega}%
^{2}(x^{\prime})\boldsymbol{,}\text{ }x^{\prime}\in\partial\Omega
\end{array}
\right.  \label{4.3}%
\end{equation}
where $a\in C_{b}(\partial\Omega)$ is a real-value function. Note that if
$a=1$ we obtain the boundary \ condition of the MIT bag model (see
\cite{Ariza} ,\cite{Ariza2}, \cite{Ariza3}, \cite{Behrndt}). Note that the
boundary condition given in the paper \cite{Behrndt}
\begin{equation}
\theta\left(  I_{4}+i\alpha_{0}(\boldsymbol{\alpha\cdot\nu)}\right)
\boldsymbol{u}_{\partial\Omega}=\left(  I_{4}+i\alpha_{0}(\boldsymbol{\alpha
\cdot\nu)}\right)  \alpha_{0}\boldsymbol{u}_{\partial\Omega} \label{4.4}%
\end{equation}
with $\theta\in C_{b}(\partial\Omega)$ can be written as%
\begin{equation}
\left\{
\begin{array}
[c]{c}%
(\theta-1)\boldsymbol{u}^{1}+(\theta+1)i(\sigma\cdot\nu)\boldsymbol{u}%
^{2}=\boldsymbol{0}\\
-(\theta-1)i(\sigma\cdot\nu)\boldsymbol{u}^{1}+(\theta+1)\boldsymbol{u}%
^{2}=\boldsymbol{0}%
\end{array}
\right.  . \label{4.4''}%
\end{equation}
The boundary condition (\ref{4.4''}) is equivalent to the condition
$\mathfrak{M}_{a}(x^{\prime})\boldsymbol{u}_{\partial\Omega}(x^{\prime
})=0,x^{\prime}\in\partial\Omega$ where
\[
a=\frac{\theta+1}{\theta-1}\in C_{b}(\partial\Omega)
\]
if
\begin{equation}
\inf_{x\in\partial\Omega}\left\vert \theta(x)-1\right\vert >0. \label{4.4'}%
\end{equation}
\ 

Note that the matrix $\mathfrak{b=i}a\left(  \sigma\cdot\boldsymbol{\nu
}\right)  $ satisfies condition (\ref{3.3'}). Hence the unbounded operator
$\mathcal{M}_{\boldsymbol{A,}\Phi,\mathfrak{M}_{a}}$ associated with
$\mathbb{M}_{\boldsymbol{A,}\Phi,\mathfrak{M}_{a}}$ is symmetric.

We consider the uniform parameter-dependent Lopatinsky-Shapiro condition for
the operator $\mathbb{D}_{\boldsymbol{A,}\Phi,\mathfrak{M}_{a}}(i\mu)$.
Applying formulas (\ref{2.16}) we obtain
\begin{align*}
\mathcal{L(}x,\xi^{\prime},\mu)  &  =(\boldsymbol{h}_{1}^{1}+ia\sigma
_{3}\boldsymbol{h}_{1}^{2},\boldsymbol{h}_{2}^{1}+ia\sigma_{3}\boldsymbol{h}%
_{2}^{2})\\
&  =(i\mu\mathbf{e}_{1}+ia\sigma_{3}\Lambda\mathbf{e}_{1},\Lambda
\boldsymbol{e}_{2}+ia\sigma_{3}i\mu\boldsymbol{e}_{2}).
\end{align*}
Formulas (\ref{2.12}), (\ref{2.13}) yield
\begin{align}
\det\mathcal{L(}x,\xi^{\prime},\mu)  &  =\det\left(
\begin{array}
[c]{c}%
i\mu+a(x)\rho\\
-ia(x)\zeta
\end{array}%
\begin{array}
[c]{c}%
\bar{\varsigma}\\
i\rho+a(x)\mu
\end{array}
\right)  =\nonumber\\
&  =\mu\rho(a^{2}(x)-1)+2ia(x)\rho^{2},x\in\partial\Omega. \label{4.6}%
\end{align}
It implies that the parameter-dependent Lopatinsky-Shapiro condition is
satisfied on $\partial\Omega$ uniformly if
\begin{equation}
\inf_{x\in\partial\Omega}\left\vert a(x)\right\vert >0. \label{4.7}%
\end{equation}
Thus Theorem \ref{te3.1} yields the following result.

\begin{theorem}
\label{te4.1} Let $\Omega\subset\mathbb{R}^{3}$ be a domain with $C^{2}%
-$uniformly regular boundary, $A_{j}\in L^{\infty}(\Omega)\boldsymbol{,}%
j=1,2,3,$ $\Phi\in L^{\infty}(\Omega)$, $a\in C_{b}(\partial\Omega)$ be
real-valued functions, and condition (\ref{4.7}) hold. Then the unbounded
operator $\mathcal{D}_{\boldsymbol{A,}\Phi,\mathfrak{M}_{a}}$ (\ref{4.3}) is
self-adjoint in $L^{2}(\Omega,\mathbb{C}^{4}).$
\end{theorem}

\begin{corollary}
\label{co4.1} Theorem \ref{te4.1} yields that the operator $\mathcal{D}%
_{\boldsymbol{A,}\Phi,\mathfrak{M}_{a}}$ where $a=\frac{1+\theta(x)}%
{1-\theta(x)}$ with $A_{j}\in L^{\infty}(\Omega)\boldsymbol{,}j=1,2,3,$
$\Phi\in L^{\infty}(\Omega)$, and
\begin{equation}
0<\inf_{x\in\partial\Omega}\left\vert a(x)\right\vert \leq\sup_{x\in
\partial\Omega}\left\vert a(x)\right\vert <\infty\label{4.7'}%
\end{equation}
is self-adjoint in $L^{2}(\Omega,\mathbb{C}^{3}).$ In particular case $a=1$ we
obtain that the operator of MIT bag model is self-adjoint.
\end{corollary}

\begin{remark}
\label{re4.1}Self-adjointness of operators of MIT bag models for domains with
\textbf{bounded} $C^{2}-$boundaries $\partial\Omega\subset\mathbb{R}^{3}$ has
been studied in the recent papers \cite{Ariza}, \cite{Ariza1}, \cite{Qurm},
\cite{Behrndt} by means of the different approach.
\end{remark}

\section{Fredholm theory and the essential spectrum}

\subsection{ Fredholmness of the operator $\mathbb{D}_{\boldsymbol{A,}%
\Phi,\mathfrak{B}}$ for bounded domains}

\begin{theorem}
\label{te5.1} Let $\Omega$ be a bounded domain with $C^{2}$-boundary,
$\boldsymbol{A}\in C^{1}(\bar{\Omega},\mathbb{C}^{3}),\Phi\in C^{1}%
(\bar{\Omega}),\mathfrak{b}_{j}\in C(\partial\Omega,\mathcal{B}(\mathbb{C}%
^{2})),j=1,2$ and the local standard Lopatinsky-Shapiro condition hold at
every point $x\in\partial\Omega.$ Then the operator $\mathbb{D}%
_{\boldsymbol{A,}\Phi,\mathfrak{B}}:H^{1}(\Omega,\mathbb{C}^{3})\rightarrow
L^{2}(\mathbb{R}^{3},\mathbb{C}^{4})$ is Fredholm.
\end{theorem}

\begin{proof}
Since $\mathfrak{D}_{\boldsymbol{A,}\Phi}$ is the elliptic operator this
theorem follows from the standard elliptic theory (see for instance
\cite{Agran}, \cite{BB}, \cite{BB1},\cite{LionsMagenes}.)
\end{proof}

\begin{corollary}
\label{co5.2}Let condition of Theorem \ref{te3.1} hold and the domain $\Omega$
is bounded with $C^{2}-$boundary. Then the operator $\mathcal{D}%
_{\boldsymbol{A,}\Phi,\mathfrak{B}}$ is self-adjoint with the discrete real spectrum.
\end{corollary}

\subsection{ Fredholmness of the operator $\mathbb{D}_{\boldsymbol{A,}%
\Phi,\mathfrak{B}}$ for unbounded domains}

Let $\chi\in C_{0}^{\infty}(\mathbb{R}^{n})$ be such that $0\leq\chi
(x)\leq1,\chi(x)=1$ for $\left\vert x\right\vert \leq1,$ and $\chi(x)=0$ for
$\left\vert x\right\vert \geq2,$ and $\chi_{R}(x)=\chi(\frac{x}{R})$,
$\psi_{R}(x)=1-\chi_{R}(x).$

\begin{definition}
\label{de5.2} Let $\Omega\subset\mathbb{R}^{3}$ be unbounded domain. (i) We
say that the operator
\[
\mathbb{D}_{\boldsymbol{A},\Phi,\mathfrak{B}}:H^{1}(\Omega,\mathbb{C}%
^{4})\rightarrow L^{2}(\Omega,\mathbb{C}^{4})
\]
\ is a locally Fredholm operator if for every $R>0\ $there exist operators
$\mathcal{L}_{R},\mathcal{R}_{R}\in\mathcal{B}(L^{2}(\Omega,\mathbb{C}%
^{4}),H^{1}(\Omega,\mathbb{C}^{4}))$ such that such that
\begin{equation}
\mathcal{L}_{R}\mathbb{D}_{\boldsymbol{A},\Phi,\mathfrak{B}}\chi_{R}I=\chi
_{R}I+T_{R}^{\prime},\chi_{R}\mathbb{D}_{\boldsymbol{A},\Phi,\mathfrak{B}%
}\mathcal{R}_{R}=\chi_{R}I+T_{R}^{\prime\prime} \label{5.12}%
\end{equation}
where $T_{R}^{\prime}\in\mathcal{K}(H^{1}(\Omega,\mathbb{C}^{4})),T_{R}%
^{\prime\prime}\in\mathcal{K}(L^{2}(\Omega,\mathbb{C}^{4})).$

(ii) We say that the operator
\[
\mathbb{D}_{\boldsymbol{A},\Phi,\mathfrak{B}}:H^{1}(\Omega,\mathbb{C}%
^{4})\rightarrow L^{2}(\Omega,\mathbb{C}^{4})
\]
\ is locally invertible at infinity if there exists $R>0\ $ and operators
$\mathcal{L}_{R}^{\prime},\mathcal{R}^{\prime}_{R}\in\mathcal{B}(L^{2}%
(\Omega,\mathbb{C}^{4}),H^{1}(\Omega,\mathbb{C}^{4}))$ such that
\begin{equation}
\mathcal{L}_{R}^{\prime}\mathbb{D}_{\boldsymbol{A},\Phi,\mathfrak{B}}\psi
_{R}I=\psi_{R}I,\psi_{R}\mathbb{D}_{\boldsymbol{A},\Phi,\mathfrak{B}%
}\mathcal{R}_{R}^{\prime}=\psi_{R}I. \label{5.13}%
\end{equation}

\end{definition}

\begin{proposition}
\label{pr5.1} \cite{Ra1}The operator
\[
\mathbb{D}_{\boldsymbol{A},\Phi,\mathfrak{B}}:H^{1}(\Omega,\mathbb{C}%
^{4})\rightarrow L^{2}(\Omega,\mathbb{C}^{4})
\]
is Fredholm if and only if $\mathbb{D}_{\boldsymbol{A},\Phi,\mathfrak{B}}$ is
locally Fredholm and locally invertible operator at infinity.
\end{proposition}

We denote by $\widetilde{\mathbb{R}^{3}}$ the compactification of
$\mathbb{R}^{3}$ obtained by the joining to the every ray $l_{\omega}=\left\{
x\in\mathbb{R}^{3}:x=t\omega,t>0,\omega\in S^{2}\right\}  $ the infinitely
distant point $\vartheta_{\omega}$. The topology in $\widetilde{\mathbb{R}%
^{3}}$ is introduced such that $\widetilde{\mathbb{R}^{3}}$ becomes
homeomorphic to the unit closed ball $\bar{B}_{1}(0).$ $\ $The fundamental
system of neighborhoods of the point $\vartheta_{\omega_{0}}$ is formed by the
conical sets $U_{\omega_{0},R}=\digamma_{\omega_{0}}\times(R,+\infty)$ where
$R>0$ and $\digamma_{\omega_{0}}$ is a neighborhood of the point $\omega_{0}$
on the unit sphere $S^{2}.$ We define the cut-off function $\varphi
_{\vartheta_{\omega_{0}}}$ of the infinitely distant point $\vartheta
_{\omega_{0}}$ as $\varphi_{\vartheta_{\omega_{0}}}=\varphi_{\omega_{0}}%
(\frac{x}{\left\vert x\right\vert })\psi_{R}(x)$ where $\varphi_{\omega_{0}%
}(\omega)\in C_{0}^{\infty}(\digamma_{\omega_{0}})$ and $\varphi_{\omega_{0}%
}(\omega)=1$ in a neighborhood $\digamma_{\omega_{0}}^{\prime}$ such that
$\overline{\digamma_{\omega_{0}}^{\prime}}\subset\digamma_{\omega_{0}}.$

If $\Omega\subset\mathbb{R}^{3}$ is an unbounded domain we denote by
$\tilde{\Omega},\widetilde{\partial\Omega}$ the closure of $\Omega
,\partial\Omega$ in $\widetilde{\mathbb{R}^{3}},$ and by $\Omega_{\infty
},\partial\Omega_{\infty}$ the associated sets of the infinitely distant points.

\begin{definition}
We say that the operator $\mathbb{D}_{\boldsymbol{A},\Phi,\mathfrak{B}}%
:H^{1}(\Omega,\mathbb{C}^{4})\rightarrow L^{2}(\Omega,\mathbb{C}^{4})$\ is
locally invertible at the\ infinitely distant point $\vartheta_{\omega}$ if
there exists $\ $a neighborhood $U_{\vartheta_{\omega}}$ of the point
$\vartheta_{\omega}$ and the operators
\[
\mathcal{L}_{\vartheta_{\omega}},\mathcal{R}_{\vartheta_{\omega}}%
\in\mathcal{B(}L^{2}(\Omega,\mathbb{C}^{4}),H^{1}(\Omega,\mathbb{C}^{4}))
\]
such that
\begin{equation}
\mathcal{L}_{\vartheta_{\omega}}\mathbb{D}_{\mathbb{D}_{\boldsymbol{A}%
,\Phi,\mathfrak{B}}}\varphi_{\vartheta_{\omega}}I=\varphi_{\vartheta_{\omega}%
}I,\varphi_{\vartheta_{\omega}}\mathbb{D}_{\mathbb{D}_{\boldsymbol{A}%
,\Phi,\mathfrak{B}}}\mathcal{R}_{\vartheta_{\omega}}=\varphi_{\vartheta
_{\omega}}I, \label{5.14}%
\end{equation}
where $\varphi_{\vartheta_{\omega}}$ is the cut-off function of the infinitely
distant point $\vartheta_{\omega}.$
\end{definition}

\begin{proposition}
\label{pr5.2}(\cite{Ra1}) The operator $\mathbb{D}_{\boldsymbol{A}%
,\Phi,\mathfrak{B}}:H^{1}(\Omega,\mathbb{C}^{4})\rightarrow L^{2}%
(\Omega,\mathbb{C}^{4})$ is a Fredholm operator if and only if $\mathbb{D}%
_{\boldsymbol{A},\Phi,\mathfrak{B}}$ is a locally Fredholm operator and
$\mathbb{D}_{\boldsymbol{A},\Phi,\mathfrak{B}}$ is locally invertible at every
infinitely distant point $\vartheta_{\omega}\in\Omega_{\infty}\cup
\partial\Omega_{\infty}$ .
\end{proposition}

\subsubsection{Fredholm property and essential spectra in the exterior of
bounded domain}

Let $\Omega\subset\mathbb{R}^{3}$ be an unbounded domain with $C^{2}-$boundary
such that $\Omega^{\prime}=\mathbb{R}^{3}\diagdown\bar{\Omega}$ is a bounded
domain. We consider the Fredholm property of the operator
\[
\mathbb{D}_{\boldsymbol{A},\Phi,\mathfrak{B}}:H^{1}(\Omega,\mathbb{C}%
^{4})\rightarrow L^{2}(\Omega,\mathbb{C}^{4})
\]
We assume as above that $A_{j},\Phi\in C_{b}^{1}(\bar{\Omega})$ and
$\mathfrak{b}_{j}\in C^{1}(\partial\Omega,\mathcal{B(}\mathbb{C}^{2})),j=1,2$.

Following to the book \cite{RRS} and the paper \ \cite{Ra1} we describe the
Fredholm properties of $\mathbb{D}_{\boldsymbol{A},\Phi,\mathfrak{B}}$ in
terms of limit operators.

We give definition of the limit operators. Let $f\in C_{b}^{1}(\mathbb{\bar
{\Omega})}$ and a sequence $\mathbb{R}^{3}\ni g_{m}\rightarrow\vartheta
_{\omega}\in\mathbb{R}_{\infty}^{3}.$ The family of functions $\left\{
f(\cdot+g_{m})\right\}  $ is uniformly bounded and equicontinuous on $\Omega.$
Then the Arcela-Ascoli Theorem yields that there exists a subsequence $h_{m}$
of $g_{m}$ and a limit function $f^{h}\in C_{b}(\bar{\Omega})$ such that
\begin{equation}
\lim_{m\rightarrow\infty}\sup_{x\in K}\left\vert f(x+h_{m})-f^{h}%
(x)\right\vert =0 \label{5.15}%
\end{equation}
for every compact set $K\subset\bar{\Omega}$.$.$

Let $h_{m}\rightarrow\vartheta_{\omega}$ be such sequence that
\[
\boldsymbol{A}(x+h_{m})\rightarrow\boldsymbol{A}^{h}(x),\Phi(x+h_{m}%
)\rightarrow\Phi^{h}(x)
\]
in the sense of the convergence defined by formula (\ref{5.15}). The operator
$\mathfrak{D}_{\boldsymbol{A,}\Phi}^{h}=\mathfrak{D}_{\boldsymbol{A}%
^{h}\boldsymbol{,}\Phi^{h}}$ is called the limit operator defined by the
sequence $h_{m},$ and we denote by $Lim_{\vartheta_{\omega}}\mathfrak{D}%
_{\boldsymbol{A,}\Phi}$ the set of all limit operators defined by the
sequences $h_{m}\rightarrow\vartheta_{\omega}.$ We set
\[
Lim\mathfrak{D}_{\boldsymbol{A,}\Phi}=%
{\displaystyle\bigcup\limits_{\vartheta_{\omega}\in\mathbb{R}_{\infty}^{3}}}
Lim_{\vartheta_{\omega}}\mathfrak{D}_{\boldsymbol{A,}\Phi}.
\]

\begin{theorem}
\label{te5.2}Let $\boldsymbol{A}\in C_{b}^{1}(\bar{\Omega},\mathbb{C}%
^{3}),\Phi\in C_{b}^{1}(\bar{\Omega}),$ $\mathfrak{b}_{j}\in C_{b}%
^{1}(\partial\Omega),j=1,2;$ $i,j=1,2,$ and the local standard
Lopatinsky-Shapiro condition hold at every point $x\in\partial\Omega.$ Then
\[
\mathbb{D}_{\boldsymbol{A},\Phi,\mathfrak{B}}:H^{1}(\Omega,\mathbb{C}%
^{4})\rightarrow L^{2}(\Omega,\mathbb{C}^{4})
\]
is the Fredholm operator if and only if all limit operators $\mathfrak{D}%
_{\boldsymbol{A,}\Phi}^{h}\in$ $Lim\mathfrak{D}_{\boldsymbol{A,}\Phi}$ are
invertible from $H^{1}(\Omega,\mathbb{C}^{4})$ into $L^{2}(\Omega
,\mathbb{C}^{4}).$
\end{theorem}

\begin{proof}
Since the boundary $\partial\Omega\subset\mathbb{R}^{3}$ is a compact surface,
the operator $\mathfrak{D}_{\boldsymbol{A,}\Phi}$ is elliptic on $\bar{\Omega
},$ and the Lopatinsky-Shapiro condition holds at every point $x\in
\partial\Omega$ the operator $\mathbb{D}_{\boldsymbol{A},\Phi,\mathfrak{B}}$
is locally Fredholm. Hence by Proposition \ref{pr5.1} $\mathbb{D}%
_{\boldsymbol{A},\Phi,\mathfrak{B}}$ is a Fredholm operator if and only if
$\mathbb{D}_{\boldsymbol{A},\Phi,\mathfrak{B}}$ is a locally invertible
operator at infinity. The operator $\mathbb{D}_{\boldsymbol{A},\Phi
,\mathfrak{B}}$ coincides with the operator $\mathfrak{D}_{\boldsymbol{A}%
,\Phi}$ outside the set $\overline{\Omega^{\prime}}.$ Applying the results of
the paper \cite{Ra2} we obtain that $\mathfrak{D}_{\boldsymbol{A},\Phi}$ is
locally invertible at infinity if and only if for every $\vartheta_{\omega}%
\in\mathbb{R}_{\infty}^{3}$ all limit operators $\mathfrak{D}_{\boldsymbol{A,}%
\Phi}^{h}\in$ $Lim_{\vartheta_{\omega}}\mathfrak{D}_{\boldsymbol{A,}\Phi}$ are
invertible from $H^{1}(\mathbb{R}^{3},\mathbb{C}^{4})$ into $L^{2}%
(\mathbb{R}^{3},\mathbb{C}^{4}).$
\end{proof}

\begin{corollary}
\label{co5.4} Let conditions of Theorem \ref{te5.2} hold. Then
\begin{equation}
sp_{ess}\mathcal{D}_{\boldsymbol{A},\Phi,\mathfrak{B}}=%
{\displaystyle\bigcup\limits_{\mathfrak{D}_{\boldsymbol{A,}\Phi}^{h}\in
Lim\mathfrak{D}_{\boldsymbol{A,}\Phi}}}
sp\mathcal{D}_{\boldsymbol{A}^{h}\boldsymbol{,}\Phi^{h}}. \label{5.16}%
\end{equation}

\end{corollary}

\begin{definition}
\label{de5.3} We say that a function $a\in C_{b}^{1}(\mathbb{R}^{3})$ is
slowly oscillating at infinity and belongs to the class $SO^{1}(\mathbb{R}%
^{3})$ if
\begin{equation}
\lim_{x\rightarrow\infty}\partial_{x_{j}}a(x)=0,j=1,2,3. \label{5.16'}%
\end{equation}
If $\Omega\subset\mathbb{R}^{3}$ is an unbounded domain we denote by
$SO^{1}(\Omega)$ the class of functions being the restrictions on $\Omega$ of
functions in $SO^{1}(\mathbb{R}^{3}).$
\end{definition}

Note that if $f\in$ $SO^{1}(\mathbb{R}^{3})$ and there exists a limit function
$f^{h}$ in the sense of formula (\ref{5.15}), then $f^{h}\in\mathbb{C}$ (see
for instance \cite{RRS}, p.228).\ 

Let $A_{j},\Phi\in SO^{1}(\Omega).$ Then the limit operators $\mathfrak{D}%
_{\boldsymbol{A,}\Phi}^{h}$ are of the form
\begin{equation}
\mathfrak{D}_{\boldsymbol{A,}\Phi}^{h}=\mathfrak{D}_{\boldsymbol{A}%
^{h}\boldsymbol{,}\Phi^{h}}=\alpha\cdot(i\boldsymbol{\nabla}+\boldsymbol{A}%
^{h})+\alpha_{0}m+\Phi^{h}I_{4} \label{5.17}%
\end{equation}
where $\boldsymbol{A}^{h}\in\mathbb{C}^{3},\Phi^{h}\in\mathbb{C}$. In the case
if $A_{j},\Phi$ are real-valued functions the operator $\mathfrak{D}%
_{\boldsymbol{A,}\Phi}^{h}$ is self-adjoint and
\begin{equation}
sp\mathfrak{D}_{\boldsymbol{A,}\Phi}^{h}=\left(  -\infty,\Phi^{h}-\left\vert
m\right\vert \right]
{\displaystyle\bigcup}
\left[  \Phi^{h}+\left\vert m\right\vert ,+\infty\right)  . \label{5.17'}%
\end{equation}
Hence formula (\ref{5.16}) yields the following result.

\begin{theorem}
\label{te5.3}Let $\Omega$ be an exterior of a bounded domain with the $C^{2}%
-$boundary $\partial\Omega,$ $\boldsymbol{A}\in SO^{1}(\Omega,\mathbb{R}%
^{3}),\Phi\in SO^{1}(\Omega,\mathbb{R})$ and the local standard
Lopatinsky-Shapiro condition is satisfied at every point $x\in\partial\Omega.$
Then
\begin{equation}
sp_{ess}\mathcal{D}_{\boldsymbol{A},\Phi,\mathfrak{B}}=(-\infty,M_{\Phi}%
^{\sup}-\left\vert m\right\vert ]%
{\displaystyle\bigcup}
[M_{\Phi}^{\inf}+\left\vert m\right\vert ,+\infty) \label{5.18}%
\end{equation}
where
\begin{equation}
M_{\Phi}^{\sup}=\limsup_{x\rightarrow\infty}\Phi(x),M_{\Phi}^{\inf}%
=\liminf_{x\rightarrow\infty}\Phi(x). \label{5.18'}%
\end{equation}

\end{theorem}

\begin{corollary}
\label{co5.3} Under conditions of Theorem \ref{te5.3}
\[
sp_{dis}\mathcal{D}_{\boldsymbol{A},\Phi,\mathfrak{B}}\subset(M_{\Phi}^{\sup
}-\left\vert m\right\vert ,M_{\Phi}^{\inf}+\left\vert m\right\vert )
\]
if $M_{\Phi}^{\sup}-M_{\Phi}^{\inf}<2\left\vert m\right\vert ,$ and if
$M_{\Phi}^{\sup}-M_{\Phi}^{\inf}\geq2\left\vert m\right\vert $ then
\end{corollary}

\[
sp\mathcal{D}_{\boldsymbol{A},\Phi,\mathfrak{B}}=sp_{ess}\mathcal{D}%
_{\boldsymbol{A},\Phi,\mathfrak{B}}=(-\infty,+\infty).
\]

\subsection{Fredholm property and essential spectrum in domains with a conical
structure at infinity}

Let $\Omega\subset\mathbb{R}^{3}$ be an connected open domain with $C^{2}%
-$boundary. We say that $\Omega$ has a conic exit at infinity if
\[
\Omega\cap B_{R}^{\prime}=\left\{  x\in\mathbb{R}^{3}:x=t\omega,t>R,\omega
\in\Sigma\right\}
\]
where $B_{R}^{\prime}=\left\{  x\in\mathbb{R}^{3}:\left\vert x\right\vert
>R\right\}  ,$ $\Sigma\subset S^{2}$ is an open set with $C^{2}$-boundary
$\partial\Sigma$. We denote by $\tilde{\Omega},\widetilde{\partial\Omega}$ the
compactifications of $\Omega,\partial\Omega$ in the topology of $\widetilde
{\mathbb{R}^{3}}.$

We consider the operator $\mathbb{D}_{\boldsymbol{A},\Phi,\mathfrak{B}}$ in
domains with a conic exit at infinity with potentials $A_{j},\Phi\in
SO^{1}(\mathbb{\Omega})$ and $\mathfrak{b}_{j}\in SO^{1}(\partial
\Omega,\mathcal{B(}\mathbb{C}^{2}))=SO^{1}(\partial\Omega)\otimes
\mathcal{B(}\mathbb{C}^{2}),j=1,2.$We define the limit operators of the
operator $\mathbb{D}_{\boldsymbol{A},\Phi,\mathfrak{B}}$ similarly to how it
was done in the articles \cite{Ra1},\cite{Ra2}.

\begin{itemize}
\item If $\vartheta_{\omega}\notin\partial\Omega_{\infty}$ then the limit
operators defined by the sequence $h_{m}\rightarrow\vartheta_{\omega}$ are the
Dirac operators $\mathfrak{D}_{\boldsymbol{A}^{h},\Phi^{h}}$ with the spectrum
given by formula (\ref{5.17'}).

\item Let $\vartheta_{\omega}\in\partial\Omega_{\infty},\ $ $l_{\omega}%
^{R}=\left\{  x\in\mathbb{R}^{3}:x=t\omega,t>R\right\}  $ and $\mathbb{T}%
_{\vartheta_{\omega}}$ be the tangent plane to $\partial\Omega$ at the$\ $ ray
$l_{\omega}^{R}$ and $\boldsymbol{\nu}(\omega)$ is the outgoing normal vector
to $\partial\Omega$ at the points of ray $l_{\omega}^{R}.$ We denote by
\[
\mathbb{R}_{+,\vartheta_{\omega}}^{3}=\left\{  y=(y^{\prime},y_{3}%
)\in\mathbb{R}^{3}:y^{\prime}\in\mathbb{T}_{\vartheta_{\omega}}y_{3}%
=t\boldsymbol{\nu}(\omega),t>0\right\}
\]
the half-space in $\mathbb{R}^{3}$ with the boundary $\mathbb{T}%
_{\vartheta_{\omega}}.$ Following to the paper \cite{Ra1} we obtain the limit
operators of $\mathbb{D}_{\boldsymbol{A},\Phi,\mathfrak{B}}$ defined by the
sequences $h_{m}\rightarrow\vartheta_{\omega}$ as
\begin{equation}
\mathbb{D}_{\boldsymbol{A},\Phi,\mathfrak{B}}^{h}\boldsymbol{u}(x)=\left\{
\begin{array}
[c]{c}%
\mathfrak{D}_{\boldsymbol{A}^{h},\Phi^{h}}\boldsymbol{u}(x),x\in
\mathbb{R}_{+,\vartheta_{\omega}}^{3},\\
\mathfrak{B}^{h}\boldsymbol{u(}s\boldsymbol{)=}\mathfrak{b}_{1}^{h}%
\boldsymbol{u}_{\mathbb{T}_{\vartheta_{\omega}}}^{1}(s)+\mathfrak{b}_{2}%
^{h}\boldsymbol{u}_{\mathbb{T}_{\vartheta_{\omega}}}^{2}(s)=\boldsymbol{0}%
,s\in\mathbb{T}_{\vartheta_{\omega}}%
\end{array}
\right.  , \label{6.1}%
\end{equation}

where
\begin{align*}
\boldsymbol{A}^{h}  &  \in\mathbb{C}^{3},\Phi^{h}\in\mathbb{C},\mathfrak{b}%
_{j}^{h}\in\mathcal{B(}\mathbb{T}_{\vartheta_{\omega}},\mathcal{B(}%
\mathbb{C}^{2})),\\
\boldsymbol{u}  &  \in H^{1}(\mathbb{R}_{+,\vartheta_{\omega}}^{3}%
,\mathbb{C}^{4}),\boldsymbol{u}_{\mathbb{T}_{\vartheta_{\omega}}}^{j}%
=\gamma_{\mathbb{T}_{\vartheta_{\omega}}}\boldsymbol{u}^{j}\in H^{1/2}%
(\mathbb{T}_{\vartheta_{\omega}},\mathbb{C}^{2}).
\end{align*}

\end{itemize}

\begin{theorem}
\label{te5.4} Let: (i) $\partial\Omega$ be a $C^{2}-$surface with a conic exit
at infinity, (ii) $A_{j},\Phi\in SO^{1}(\mathbb{\Omega})$ and $\mathfrak{b}%
_{j}\in SO^{1}(\partial\Omega,\mathcal{B(}\mathbb{C}^{2}))=SO^{1}%
(\partial\Omega)\otimes\mathcal{B(}\mathbb{C}^{2}),$ (iii) the
Lopatinsky-Shapiro condition be satisfied at every point $x\in\partial\Omega.$
Then the operator $\mathfrak{D}_{\boldsymbol{A},\Phi,\mathfrak{B}}%
:H^{1}(\mathbb{\Omega},\mathbb{C}^{4})\rightarrow L^{2}(\Omega,\mathbb{C}%
^{4})$ is a Fredholm operator if and only if for every $\vartheta_{\omega}%
\in\partial\Omega_{\infty}$ all limit operators $\mathbb{D}_{\boldsymbol{A}%
^{h},\Phi^{h},\mathfrak{B}^{h}}$ $\in Lim_{\vartheta_{\omega}}\mathfrak{D}%
_{\boldsymbol{A},\Phi,\mathfrak{B}}$ defined by formulas (\ref{5.17}),
(\ref{6.1}) are invertible.
\end{theorem}

\begin{proof}
The operator $\mathbb{D}_{\boldsymbol{A},\Phi,\mathfrak{B}}$ is locally
Fredholm since the Dirac operator $\mathfrak{D}_{A,\Phi}$ is elliptic and the
Lopatinsky condition holds at every point $x\in\partial\Omega.$ According
Proposition \ref{p3.1} $\mathbb{D}_{\boldsymbol{A},\Phi,\mathfrak{B}}$ is a
Fredholm operator if and only if $\mathbb{D}_{\boldsymbol{A},\Phi
,\mathfrak{B}}$ is locally invertible at every infinitely distant point
$\vartheta_{\omega}\in\bar{\Omega}_{\infty}$. Following to the monograph
\cite{RRS}, and the paper \cite{Ra1},\cite{Ra2} we obtain the statement of
Theorem \ref{te5.4}.
\end{proof}

\begin{corollary}
\label{co5.5} Let the conditions of Theorem \ref{te5.4} hold. Then
\begin{equation}
sp_{ess}\mathcal{D}_{\boldsymbol{A},\Phi,\mathfrak{B}}=%
{\displaystyle\bigcup\limits_{\mathcal{D}_{\boldsymbol{A},\Phi,\mathfrak{B}%
}^{h}\in Lim\mathcal{D}_{\boldsymbol{A},\Phi,\mathfrak{B}}}}
sp\mathcal{D}_{\boldsymbol{A},\Phi,\mathfrak{B}}^{h} \label{6.2}%
\end{equation}
where $\mathcal{D}_{\boldsymbol{A},\Phi,\mathfrak{B}}^{h}$ are unbounded
operators associated with the limit operators $\mathbb{D}_{\boldsymbol{A}%
,\Phi,\mathfrak{B}}^{h}.$
\end{corollary}

\subsection{Essential spectrum of the operator of MIT bag model in domains
with conic exit to infinity}

We consider the operator of MIT bag model in the domain $\Omega$ with $C^{2}%
-$conical at infinity boundary
\begin{equation}
\mathbb{M}_{\boldsymbol{A,}\Phi,\mathfrak{M}}\boldsymbol{u}(x)\boldsymbol{=}%
\left\{
\begin{array}
[c]{c}%
\mathfrak{D}_{\boldsymbol{A},\Phi}\boldsymbol{u}(x),x\in\Omega\\
\mathfrak{M}\boldsymbol{u}=\mathbf{u}_{\partial\Omega}^{1}+i\left(
\sigma\cdot\boldsymbol{\nu}\right)  \boldsymbol{u}_{\partial\Omega}%
^{2}\boldsymbol{=0}\text{ on }\partial\Omega
\end{array}
\right.  . \label{7.1}%
\end{equation}
We assume that the potentials $\boldsymbol{A}\in SO^{1}(\Omega,\mathbb{R}%
^{3}),\Phi\in SO^{1}(\Omega,\mathbb{R})$ are real-valued. Since the conical at
infinity $C^{2}-$surface is uniformly regular the unbounded operator
$\mathcal{M}_{\boldsymbol{A,}\Phi,\mathfrak{M}}$ associated with
$\mathbb{M}_{\boldsymbol{A,}\Phi,\mathfrak{M}}:H^{1}(\Omega,\mathbb{C}%
^{4})\rightarrow L^{2}(\Omega,\mathbb{C}^{4})$ is self-adjoint in
$L^{2}(\Omega,\mathbb{C}^{4})$ and $sp_{ess}\mathcal{M}_{\boldsymbol{A,}%
\Phi,\mathfrak{M}}$ is defined by formula (\ref{6.2}). That is
\[
sp_{ess}\mathcal{M}_{\boldsymbol{A,}\Phi,\mathfrak{M}}=%
{\displaystyle\bigcup\limits_{\mathcal{M}_{\boldsymbol{A},\Phi,\mathfrak{M}%
}^{h}\in Lim\mathcal{M}_{\boldsymbol{A},\Phi,\mathfrak{M}}}}
sp\mathcal{M}_{\boldsymbol{A},\Phi,\mathfrak{M}}^{h},
\]
where $\mathcal{M}_{\boldsymbol{A},\Phi,\mathfrak{B}}^{h}$ are limit operators
of the operator $\mathcal{M}_{\boldsymbol{A},\Phi,\mathfrak{M}}.$

Let the sequence $h_{m}\rightarrow\vartheta_{\omega}\in\Omega_{\infty
}\diagdown\partial\Omega_{\infty}.$ Then the limit operators are of the form
$\mathcal{M}_{\boldsymbol{A,}\Phi,\mathfrak{M}}^{h}=\mathfrak{D}%
_{\boldsymbol{A}^{h},\Phi^{h}}$ and
\begin{equation}
sp\mathcal{M}_{\boldsymbol{A,}\Phi,\mathfrak{M}}^{h}=sp\mathfrak{D}%
_{\boldsymbol{A}^{h},\Phi^{h}}=\left(  -\infty,\Phi^{h}-\left\vert
m\right\vert \right]  \cup\left[  \Phi^{h}+\left\vert m\right\vert
,+\infty\right)  . \label{7.1'}%
\end{equation}

Let the sequence $h_{m}\rightarrow\vartheta_{\omega}\in\partial\Omega_{\infty
}.$ Then without loss of generality we assume that $\mathbb{T}_{\vartheta
_{\omega}}=\mathbb{R}_{x^{\prime}}^{2}=\left\{  x=(x^{\prime},x_{3}%
):x_{3}=0\right\}  .$ Hence the limit operators of $\mathbb{M}%
_{\boldsymbol{A,}\Phi}$ are of the form%
\begin{equation}
\mathbb{M}_{\boldsymbol{A,}\Phi,\mathfrak{M}_{\vartheta_{\omega}}}%
^{h}\boldsymbol{u}(x)\boldsymbol{=}\left\{
\begin{array}
[c]{c}%
\mathfrak{D}_{\boldsymbol{A}^{h},\Phi^{h}}\boldsymbol{u}(x),x\in\mathbb{R}%
_{+}^{3}\\
\mathfrak{M}_{\vartheta_{\omega}}\boldsymbol{u}(x^{\prime})=\boldsymbol{u}%
(x^{\prime},+0)+i\sigma_{3}\boldsymbol{u}^{2}(x^{\prime},+0)\boldsymbol{=0}%
\text{ on }\mathbb{R}_{x^{\prime}}^{2}%
\end{array}
\right.  . \label{7.2}%
\end{equation}

The gauge transformation $x\rightarrow e^{-i\boldsymbol{A}^{h}\cdot x}x$
reduces the study the spectrum of $\mathbb{M}_{\boldsymbol{A,}\Phi
,\mathfrak{G}}^{h}$ to the spectrum of the operator $\mathbb{M}%
_{\boldsymbol{0,}\Phi^{h},\mathfrak{G}}=\mathbb{M}_{\boldsymbol{0,}%
0,\mathfrak{G}}+\Phi^{h}I_{4}$ where
\begin{equation}
\mathbb{M}_{\boldsymbol{0,}0,\mathfrak{G}}\boldsymbol{u=}\left\{
\begin{array}
[c]{c}%
(i\alpha\cdot\boldsymbol{\nabla}+\alpha_{0}m)\boldsymbol{u}(x),x\in
\mathbb{R}_{+}^{3}\\
\mathfrak{G}\boldsymbol{u}(x^{\prime})=\boldsymbol{u}(x^{\prime}%
,+0)+i\sigma_{3}\boldsymbol{u}^{2}(x^{\prime},+0)\boldsymbol{=0,}x^{\prime}%
\in\mathbb{R}_{x^{\prime}}^{2}%
\end{array}
\right.  . \label{7.2'}%
\end{equation}

After the Fourier transform in (\ref{7.2'}) with respect to $x^{\prime}%
\in\mathbb{R}^{2}$ we obtain the family of one-dimensional Dirac operators
depending on the parameter $\xi^{\prime}\in\mathbb{R}^{2}$
\begin{equation}
\mathbb{L}(\xi^{\prime})\boldsymbol{v}(z)=\left\{
\begin{array}
[c]{c}%
\left(  \alpha^{\prime}\cdot\xi^{\prime}+i\alpha_{3}\frac{d}{dz}+\alpha
_{0}m\right)  \boldsymbol{v}(z),z\in\mathbb{R}_{+},\xi^{\prime}\in
\mathbb{R}^{2}\\
\boldsymbol{v}^{1}(+0)+i\sigma_{3}\boldsymbol{v}^{2}(+0)\boldsymbol{=0}.
\end{array}
\right.  . \label{6.3}%
\end{equation}
The operator $\mathbb{L}(\xi^{\prime})$ has the essential spectrum
\begin{equation}
sp_{ess}\mathbb{L}(\xi^{\prime})=\left(  -\infty,-\sqrt{\left\vert \xi
^{\prime}\right\vert ^{2}+\left\vert m\right\vert ^{2}}\right]
{\displaystyle\bigcup}
\left[  \sqrt{\left\vert \xi^{\prime}\right\vert ^{2}+\left\vert m\right\vert
^{2}},+\infty\right)  \label{6.4''}%
\end{equation}
and a possible discrete spectrum
\begin{equation}
sp_{dis}\mathbb{L}(\xi^{\prime})\subset\left(  -\sqrt{\left\vert \xi^{\prime
}\right\vert ^{2}+m^{2}},\sqrt{\left\vert \xi^{\prime}\right\vert ^{2}+m^{2}%
}\right)  . \label{6.4'''}%
\end{equation}
We are looking for $sp_{dis}\mathbb{L}(\xi^{\prime})$ as follows. The
equation
\begin{equation}
\left(  \alpha^{\prime}\cdot\xi^{\prime}+i\alpha_{3}\frac{d}{dz}+\alpha
_{0}m-\lambda I_{4}\right)  \boldsymbol{v}(z)=0,z>0,\lambda\in\mathbb{R}
\label{6.4'}%
\end{equation}
has the exponentially decreasing solutions of the form $\boldsymbol{v}%
(z)=\boldsymbol{h}e^{-\rho z}$,$\rho=\sqrt{\left\vert \xi^{\prime}\right\vert
^{2}+m^{2}-\lambda^{2}}>0\ $\ with $\left\vert \lambda\right\vert
<\sqrt{\left\vert \xi^{\prime}\right\vert ^{2}+m^{2}},$ where the vector
$\boldsymbol{h}$ satisfies the equation
\begin{equation}
\Theta_{-}\left(  \boldsymbol{\xi}^{\prime},\lambda\right)  \boldsymbol{h}%
=\left(  \alpha^{\prime}\cdot\xi^{\prime}-i\rho\alpha_{3}+\alpha_{0}m-\lambda
I_{4}\right)  \boldsymbol{h=0.} \label{6.5}%
\end{equation}
Equation (\ref{6.5}) has the general solution
\begin{equation}
h=\boldsymbol{h}(\xi^{\prime},\lambda)=\Theta_{+}\left(  \boldsymbol{\xi
}^{\prime},\lambda\right)  \boldsymbol{f} \label{6.6}%
\end{equation}
where $\boldsymbol{f}$ $\in\mathbb{C}^{4}$ is an arbitrary vector and
\[
\Theta_{+}\left(  \boldsymbol{\xi}^{\prime},\lambda\right)  =\alpha^{\prime
}\cdot\boldsymbol{\xi}^{\prime}-i\rho\mathfrak{\alpha}_{3}+\alpha_{0}m+\lambda
I_{4}=\left(
\begin{array}
[c]{cc}%
\left(  \lambda+m\right)  I_{2} & \Lambda(\boldsymbol{\xi}^{\prime},\lambda)\\
\Lambda(\boldsymbol{\xi}^{\prime},\lambda) & \left(  \lambda-m\right)  I_{2}%
\end{array}
\right)  ,
\]
and%
\begin{equation}
\Lambda=\Lambda(\boldsymbol{\xi}^{\prime},\lambda)=\sigma^{\prime}%
\cdot\boldsymbol{\xi}^{\prime}-i\rho\sigma_{3}=\left(
\begin{array}
[c]{cc}%
-i\rho & \bar{\varsigma}\\
\varsigma & i\rho
\end{array}
\right)  ,\varsigma=\xi_{1}+i\xi_{2}. \label{6.7}%
\end{equation}
Let $\boldsymbol{h}_{1},\boldsymbol{h}_{2}$ be two linear independent
solutions of equation (\ref{6.4'}) of the form%

\begin{equation}
\boldsymbol{h}_{1}=\boldsymbol{h}_{1}(\boldsymbol{\xi}^{\prime},\lambda
)=\Theta(\boldsymbol{\xi}^{\prime},\lambda)\left(
\begin{array}
[c]{c}%
\boldsymbol{e}_{1}\\
\boldsymbol{0}%
\end{array}
\right)  =\left(
\begin{array}
[c]{c}%
\left(  \lambda+m\right)  \boldsymbol{e}_{1}\\
\Lambda(\boldsymbol{\xi}^{\prime},\lambda)\boldsymbol{e}_{1}%
\end{array}
\right)  , \label{6.8}%
\end{equation}

\begin{equation}
\boldsymbol{h}_{2}=\boldsymbol{h}_{2}(\boldsymbol{\xi}^{\prime},\lambda
)=\Theta(\boldsymbol{\xi}^{\prime},\lambda)\left(
\begin{array}
[c]{c}%
0\\
\boldsymbol{e}_{2}%
\end{array}
\right)  =\left(
\begin{array}
[c]{c}%
\Lambda(\boldsymbol{\xi}^{\prime},\lambda)\boldsymbol{e}_{2}\\
\left(  \lambda-m\right)  \boldsymbol{e}_{2}%
\end{array}
\right)  \label{6.8'}%
\end{equation}
where
\begin{equation}
\Lambda(\boldsymbol{\xi}^{\prime},\lambda)\boldsymbol{e}_{1}=\left(
\begin{array}
[c]{c}%
-i\rho\\
\varsigma
\end{array}
\right)  ,\Lambda(\boldsymbol{\xi}^{\prime},\lambda)\boldsymbol{e}_{2}=\left(
\begin{array}
[c]{c}%
\bar{\varsigma}\\
i\rho
\end{array}
\right)  . \label{6.9}%
\end{equation}
Then the general solution of equation (\ref{6.4'}) is
\begin{equation}
\boldsymbol{v}(z)=C_{1}\boldsymbol{h}_{1}e^{-\rho z}+C_{2}\boldsymbol{h}%
_{2}e^{-\rho z},z>0. \label{6.10}%
\end{equation}
Substituting (\ref{6.10}) in the boundary condition $\boldsymbol{v}%
^{1}(+0)+i\sigma_{3}\boldsymbol{v}^{2}(+0)\boldsymbol{=0}$ we obtain the
system of linear equations with respect to $C_{1},C_{2}$
\begin{equation}
C_{1}(\boldsymbol{h}_{1}^{1}+i\sigma_{3}\boldsymbol{h}_{1}^{2})+C_{2}%
(\boldsymbol{h}_{2}^{1}+i\sigma_{3}\boldsymbol{h}_{2}^{2})=0. \label{6.11}%
\end{equation}

\bigskip System (\ref{6.11}) has a nontrivial solution if and only if
\begin{align}
&  \det\left(  \boldsymbol{h}_{1}^{1}+i\sigma_{3}\boldsymbol{h}_{1}%
^{2},\boldsymbol{h}_{2}^{1}+i\sigma_{3}\boldsymbol{h}_{2}^{2}\right)
\label{6.12}\\
&  =\det\left(
\begin{array}
[c]{cc}%
\lambda+m+\rho & \bar{\varsigma}\\
-i\varsigma & i\rho-i(\lambda-m)
\end{array}
\right)  =2i\rho(\rho+m)=0.\nonumber
\end{align}

We consider two cases:

1) The mass of the particle $m\geq0.$ In this case equation (\ref{6.12}) does
not have positive solutions and $sp\mathcal{M}_{\boldsymbol{A,}\Phi
,\mathfrak{G}}^{h}$ is given by formula (\ref{7.1'}). It implies that
$sp_{ess}\mathcal{M}_{\boldsymbol{A,}\Phi,\mathfrak{G}}$ is given by formulas
(\ref{5.18}),(\ref{5.18'}).

2) $\ $Let the mass $m<0.$ Then the equation (\ref{6.12}) has the positive
solutions $\rho=\left\vert m\right\vert .$ It follows from equation%
\[
\rho=\sqrt{\left\vert \xi^{\prime}\right\vert ^{2}+m^{2}-\lambda^{2}%
}=\left\vert m\right\vert
\]
that the operator $\mathbb{L}^{0}(\xi^{\prime})$ has eigenvalues $\lambda
_{\pm}(\xi^{\prime})=\pm\left\vert \xi^{\prime}\right\vert ,$ for every
$\xi^{\prime}\in\mathbb{R}^{2}.$ It yields that
\begin{equation}
sp\mathcal{M}_{\boldsymbol{A}^{h}\boldsymbol{,}\Phi^{h},\mathfrak{G}}=\left(
-\infty,+\infty\right)  , \label{6.13}%
\end{equation}
Taking in the account (\ref{6.2}),(\ref{6.13}) we obtain that
\[
sp\mathcal{D}_{\boldsymbol{A},\Phi,\mathfrak{M}}=sp_{ess}\mathcal{D}%
_{\boldsymbol{A},\Phi,\mathfrak{M}}=(-\infty,+\infty).
\]

Thus, we established the following result.

\begin{theorem}
\label{te7.1}Let $\boldsymbol{A}\in SO^{1}(\Omega,\mathbb{R}^{3}),$ and
$\Phi\in SO^{1}(\Omega)$ be the real valued functions and $\Omega$ is a domain
with a $C^{2}-$boundary having a conic exit at infinity. Then the essential
spectrum of MIT bag model is defined by the formulas:

(a) If $m\geq0$
\begin{equation}
sp_{ess}\mathcal{M}_{\boldsymbol{A},\Phi,\mathfrak{M}}=(-\infty,M_{\Phi}%
^{\sup}-\left\vert m\right\vert ]%
{\displaystyle\bigcup}
[M_{\Phi}^{\inf}+\left\vert m\right\vert ,+\infty)
\end{equation}
where
\begin{equation}
M_{\Phi}^{\sup}=\limsup_{x\rightarrow\infty}\Phi(x),M_{\Phi}^{\inf}%
=\liminf_{x\rightarrow\infty}\Phi(x).
\end{equation}

(b) If $m<0$
\[
sp_{ess}\mathcal{M}_{\boldsymbol{A},\Phi,\mathfrak{M}}=(-\infty,+\infty).
\]

\end{theorem}

It follows from Theorem \ref{te7.1} that the operator $\mathcal{M}%
_{\boldsymbol{A},\Phi,\mathfrak{M}}$ may have the discrete spectrum on the
interval $\left(  M_{\Phi}^{\sup}-\left\vert m\right\vert ,M_{\Phi}^{\inf
}+\left\vert m\right\vert \right)  $ if $m>0$ and $M_{\Phi}^{\sup}-M_{\Phi
}^{\inf}<2m,$ and $\mathcal{M}_{\boldsymbol{A},\Phi,\mathfrak{M}}$ does not
have discrete spectrum if $m<0.$

\begin{acknowledgement}
This work was supported by the Mexican National System of Investigators (SNI).
\end{acknowledgement}

\bigskip

\bigskip

\bigskip

\bigskip
\end{document}